\documentclass{article}

\usepackage{microtype}
\usepackage{graphicx}
\usepackage{subcaption}
\usepackage{booktabs}

\usepackage{hyperref}

\usepackage{algorithm,algorithmic}
\usepackage{breakcites}

\usepackage{amsmath}
\usepackage{amssymb}
\usepackage{mathtools}
\usepackage{amsthm}

\usepackage{ifthen}
\usepackage{xargs}

\usepackage[capitalize,noabbrev]{cleveref}
\theoremstyle{plain}
\newtheorem{theorem}{Theorem}[section]
\newtheorem{proposition}[theorem]{Proposition}
\newtheorem{lemma}[theorem]{Lemma}

\theoremstyle{definition}
\newtheorem{definition}[theorem]{Definition}
\newtheorem{assumption}[theorem]{Assumption}
\theoremstyle{remark}
\newtheorem{remark}[theorem]{Remark}

\usepackage[textsize=tiny]{todonotes}

\usepackage[left=2cm,right=2cm,top=2cm,bottom=2cm]{geometry}
\usepackage{stmaryrd}
\usepackage{authblk}

\newcommandx{\amap}[2][1=]{\ifthenelse{\equal{#1}{}}{\mathcal{A}^{#2}}{\mathcal{A}^{#2}_{#1}}}
\newcommandx{\hmap}[2][1=]{\ifthenelse{\equal{#1}{}}{\mathcal{H}^{#2}}{\mathcal{H}^{#2}_{#1}}}
\newcommandx{\vmap}[2][1=]{\ifthenelse{\equal{#1}{}}{\mathcal{V}^{#2}}{\mathcal{V}^{#2}_{#1}}}
\newcommandx{\umap}[2][1=]{\ifthenelse{\equal{#1}{}}{\mathcal{U}^{#2}}{\mathcal{U}^{#2}_{#1}}}
\newcommand{\oscnorm}{\operatorname{osc}}
\newcommand{\bmf}{\mathsf{F}_{\mathsf{b}}}
\newcommand{\bmffd}{\mathcal{F}_{\mathsf{b}}}
\newcommand{\mf}{\mathsf{F}}
\newcommand{\tensprod}{\varotimes}
\newcommand{\lip}{\mathcal{L}}
\newcommand{\liptd}{\tilde{\mathcal{L}}}
\newcommand{\Xsp}{\mathsf{X}}
\newcommand{\Xfd}{\mathcal{X}}
\newcommand{\Ysp}{\mathsf{Y}}
\newcommand{\Yfd}{\mathcal{Y}}
\newcommand{\Asp}{\mathsf{A}}
\newcommand{\Afd}{\mathcal{A}}
\newcommand{\Zsp}{\mathsf{Z}}
\newcommand{\Zfd}{\mathcal{Z}}

\newcommand{\intset}[2]{\llbracket #1, #2 \rrbracket}
\newcommand{\nsetpos}{\mathbb{N}_{>0}}
\renewcommand{\geqslant}{\geq}

\newcommand{\emk}{T^{\thp, \parvar}}
\newcommand{\COLBO}{\ell}
\newcommand{\contr}{\lambda}
\newcommand{\eqdef}{\vcentcolon=}

\newcommand{\thp}{\theta}
\newcommand{\thpsp}{\Theta}
\newcommand{\statedim}{d_x}
\newcommand{\obsdim}{d_y}

\newcommand{\vparsp}{\Phi}
\newcommand{\statesp}{\mathbb{R}^{\statedim}}
\newcommand{\obssp}{\mathbb{R}^{\obsdim}}

\newcommand{\nrmbackwdweight}[3][\parvar]{\bar{w}_{#2-1|#2}^{#1, #3}}
\newcommand{\sample}[2]{\xi_{#1}^{#2}}

\newcommand{\afelbo}[1]{\ell_{#1}^{\thp,\parvar}}
\newcommand{\addfelbo}[2][\parvar]{\ell_{#2}^{\, \thp, #1}}

\newcommand{\grad}{\nabla}
\newcommandx{\gdiffELBO}[2][1=]{
\ifthenelse{\equal{#1}{}}
{\mathcal{G}_{#2}^{\thp, \parvar}}
{\mathcal{G}_{#2}^{\thp, \parvar}\langle #1 \rangle}
}

\newcommand{\gradalt}[1]{\nabla_{#1}}
\newcommand{\parvec}{\the}
\newcommand{\parspace}{\the}
\newcommand{\parvar}{\phi}
\newcommand{\md}[1]{g_{#1}^\thp}
\newcommand{\hd}[1]{m_{#1}^\thp}
\newcommand{\pE}{\mathbb{E}}
\newcommand{\postletter}{\pi}
\newcommandx\post[2][1=]{
\ifthenelse{\equal{#1}{}}
	{\ensuremath{\postletter_{#2}^\thp}}
	{\ensuremath{\postletter_{#2}^{\thp, \N}}}
}
\newcommand{\rset}{\ensuremath{\mathbb{R}}}
\newcommand{\rsetpos}{\ensuremath{\mathbb{R}_{>0}}}
\newcommand{\rmd}{\ensuremath{\mathrm{d}}}
\newcommand{\eqsp}{}
\newcommand{\nset}{\mathbb{N}}
\newcommand{\Km}{K^{\thp, \parvar}}
\newcommand{\K}{T^{\thp, \parvar}}
\renewcommand{\S}{S^{\parvar}}
\newcommand{\R}{R}

\newcommand{\mixrt}{\alpha}
\newcommand{\probmeas}{\mathsf{M}_1}
\newcommand{\jointd}[1]{p_{#1}^{\thp}}
\newcommand{\vd}[2][\parvar]{q_{#2}^{#1}} 
\newcommand{\vk}[2][\parvar]{Q_{#2}^{#1}}

\newcommand{\expect}[2]{\mathbb{E}_{#1}\left[#2\right]}
\newcommand{\approxexpect}[2]{\widehat{\mathbb{E}}_{#1}{\left[#2\right]}}
\newcommand{\ELBO}[2][\parvar]{\mathcal{L}_{#2}^{\thp, #1}}

\newcommand{\gELBOt}[3]{\grad_{#1}\mathcal{L}_{#2}^{#3}}
\newcommand{\approxELBO}[2][\parvar]{\widehat{\mathcal{L}}^{\thp, #1}_{#2}}
\newcommand{\hGELBO}[2]{\widehat{\grad}_{#1} \mathcal{L}_{#2}^{\thp,\parvar}}
\newcommand{\hGELBOt}[3]{\widehat{\grad}_{#1} \mathcal{L}_{#2}^{\thp_{#3},\parvar_{#3}}}
\newcommand{\Hstat}[1]{h_{#1}}
\newcommand{\Gstat}[1]{u_{#1}}
\newcommand{\Fstat}[1]{v_{#1}}
\newcommand{\backwd}[1]{b_{#1}^\thp}
\newcommand{\fwdpot}[2][\parvar]{\psi_{#2}^{#1}}
\newcommand{\approxHstat}[3][\parvar]{\hat{h}_{#2}^{#1, #3}}
\newcommand{\approxGstat}[3][\parvar]{\hat{u}_{#2}^{#1, #3}}
\newcommand{\approxFstat}[3][\thp]{\hat{v}_{#2}^{#1, #3}}

\newcommand{\paramapproxvarcondE}[2][H]{#1_{\hat{\gamma}_{#2}}^\parvar}
\newcommand{\paramvarcondE}{H_{\gamma}^\parvar}

\newcommand{\fat}[1][]{f^{\parvar_{#1}}}
\newcommand{\fxt}[1][]{\tilde{f}^{\parvar_{#1}}}
\newcommand{\statemap}{\zeta}
\newcommand{\RMCVI}{RMCVI}

\newcommand{\X}{\mathsf{X}}

\newcommand{\Z}{\mathsf{Z}}

\title{Efficient Online Variational Estimation via Monte Carlo Sampling}
\date{}

\author[$\dag$]{Mathis Chagneux}
\author[$\ddag$]{Mathias Muller}
\author[$\star$]{ Pierre Gloaguen}
\author[$\top$]{Sylvain Le Corff}
\author[$\ddag$]{Jimmy Olsson}

\affil[$\dag$]{{\small T\'el\'ecom Paris, Institut Polytechnique de Paris, Palaiseau.}}
\affil[$\ddag$]{{\small KTH Royal Institute ofTechnology.}}
\affil[$\star$]{{\small Universit\'e Bretagne Sud, UMR CNRS 6205, LMBA, F-56000 Vannes, France.}}
\affil[$\top$]{{\small Sorbonne Universit\'e, Universit\'e Paris Cit\'e, CNRS, Laboratoire de Probabilit\'es, Statistique et Mod\'elisation, LPSM, F-75005 Paris, France.}}

\begin{document}

\maketitle

\begin{abstract}
  This article addresses online variational estimation in parametric state-space models. We propose a new procedure for efficiently computing the evidence lower bound and its gradient in a streaming-data setting, where observations arrive sequentially.
The algorithm allows for the simultaneous training of the model parameters and the distribution of the latent states given the observations.
It is based on i.i.d. Monte Carlo sampling, coupled with a well-chosen deep architecture, enabling both computational efficiency and flexibility. The performance of the method is illustrated on both synthetic data and real-world air-quality data.
The proposed approach is theoretically motivated by the existence of an asymptotic contrast function and the ergodicity of the underlying Markov chain, and applies more generally to the computation of additive expectations under posterior distributions in state-space models.
\end{abstract}

\section{Introduction}
\label{sec:intro}

This work considers \emph{state-space models} (SSMs) where the law of observations $(Y_t)_{t \in \nset}$ is governed by an unobserved, or `hidden', Markov chain $(X_t)_{t \in \nset}$, and the finite dimensional distributions of $(X_t, Y_t)_{t \in \nset}$ are given by parametric distributions indexed by $\thp \in \Theta$. Learning the parameter $\theta$ in this context is a complex task, as it usually requires access to the \textit{joint-smoothing distributions}, defined at time $t$ as the conditional distribution of the latent states $X_{0:t}$\footnote{$a_{u:v}$ is a short-hand notation for $(a_u,\ldots,a_v)$.} given the corresponding observations $Y_{0:t}$. 
This paper addresses the challenging problem of \emph{online estimation} in SSMs, which consists of sequentially learning both $\thp$ and the smoothing distributions as data stream in real time. 

A first possible approach to this task is based on sequential Monte Carlo (SMC) smoothing (see \cite{chopin2020introduction} and the references therein). 
Such SMC-based methods for online learning come with strong theoretical guaranties \cite{le2013convergence,olsson2017efficient,tadic2020asymptotic,gao2025parameter}, but typically suffer from the curse of dimensionality in the state dimension. This recently motivated the use of variational inference (VI) \cite{blei2017} as an alternative, as VI has demonstrated strong empirical performance for online inference in large-scale settings, particularly through stochastic variational frameworks  \cite{hoffman2013stochastic, broderick2013streaming}. 

 Taking a variational approach, the smoothing distributions are approximated by simpler distributions, the \emph{variational distributions}, depending on some unknown parameter $\parvar\in\vparsp$.  Both $\thp$ and $\parvar$ are then learned by maximizing a proxy of the log likelihood, the \emph{evidence lower bound} (ELBO), which is generally done using gradient ascent, thus requiring the computation of the ELBO's gradient. 
A key feature of these approaches is the choice of the variational distribution, which must be suited for online learning. \cite{campbell2021online} rely on structured variational distributions that both mimic the Markovian form of the  smoothing distributions (see \Cref{sec:model:background}) and are well suited to online learning. 
The authors explore online variational additive smoothing for the recursive computation of the ELBO and its gradients, using a Bellman-type recursion inspired by reinforcement learning and recursive maximum likelihood methods.  A major limitation of this algorithm is that, \emph{at each iteration $t$}, it requires solving an inner optimization problem to learn a regression function serving as a proxy for a conditional expectation. This step can be computationally expensive and depends critically on the appropriate choice of the regression class.  
More broadly, a key drawback of existing online variational learning methods is the lack of theoretical guaranties for the proposed algorithms. 

In this paper, we develop a theoretically grounded framework for online variational learning. 
Our approach is rooted in a stochastic approximation perspective, and, following \cite{mastrototaro:muller:olsson:2025}, we show that for a structured variational distributions parameterized by flexible function classes (such as deep neural networks), online variational learning amounts to maximizing a lower bound (COLBO) on the limiting time-normalized asymptotic log likelihood, also known as the asymptotic contrast function. Building on this theoretical framework, we first formulate an ideal---but generally intractable---stochastic approximation  algorithm that maximizes the COLBO, and then propose a Monte Carlo version of then same, using an efficient importance sampling approach that avoids any regression task and outperforms the algorithm of \cite{campbell2021online} in terms of computation time.
The main contributions of this paper can be summarized as follows:

\begin{itemize}
    \item We provide a theoretically grounded framework for online variational learning, showing that it can be viewed as Robbins---Monro algorithm.
    \item We propose a computationally efficient online estimator of the COLBO its gradient in the context of SSMs. In contrast to computationally intensive SMC or Markov chain Monte Carlo (MCMC) methods, our algorithm, which we refer to as \emph{Recursive Monte Carlo Variational Inference} (\RMCVI) to emphasize its iterative structure, relies on simple i.i.d. samples from the marginal variational distributions. 
     \item The proposed algorithm is not limited to the online computation and optimization of the COLBO and can be directly adapted to compute more general expectations, training losses, or gradients under distributions that admit a Markovian structure. 
    \item Experimentally, we demonstrate the performance of our estimator  both on synthetic and real world data.
\end{itemize}

\section{Related work}
Our methodology draws on recent advances in smoothing  methods for SSMs by (i) proposing a Monte Carlo approach for approximating conditional expectations  and (ii) relying on a structured variational family.

\emph{SMC for online learning in SSMs.} The original approaches to online smoothing is based on SMC methods. 
We refer the reader to \cite[Section~11]{douc2014nonlinear} for a presentation of the general concepts underlying SMC algorithms, and to \cite{olsson2017efficient,gloaguen2022pseudo,dau2022complexity} for more recent developments and applications of these methods to online smoothing.  
Theoretical guaranties for online optimization of the log-likelihood using SMC can be found in \cite{le2013convergence,olsson2017efficient,tadic2020asymptotic,gao2025parameter}.

\emph{Offline VI for SSMs.} Early works on VI for SSMs primarily focus on offline estimation \cite{johnson2016, krishnan2017structured, lin2018variational}, \emph{i.e.}, they require prior access to the entire observation sequence $Y_{0:t}$ in order to compute the gradients of the ELBO. These approaches rely on a \textit{forward} factorization of the variational distribution, which is incompatible with online learning. 

\emph{Online VI for SSMs.} \cite{Marino2018AGM, zhao2020variational, dowling2023real} opts to trade smoothing for filtering by targeting the marginal distributions at each timestep with variational distributions that depend only on observations up to $t$. 
By designing a new variational family, \cite{campbell2021online} show that the ELBO and its gradient can be recursively expressed via conditional expectations, providing a natural framework for online learning.  
The authors approximate these conditional expectations at each time step using functional approximations. 

\emph{Online Variational SMC.} A recent approach to improving SMC for online learning is to adapt the particle proposal dynamics by optimizing a variational objective \cite{zhao2022streaming, mastrototaro:olsson:2024, mastrototaro:muller:olsson:2025}. In this framework, the variational method aims to approximate the locally optimal proposal while simultaneously learning the model parameters, rather than targeting the full joint smoothing distribution. 

\emph{Theoretical guaranties for VI.} In the offline setting, \cite{chagneux2022amortized} established the first theoretical results on error control in VI for SSM, building on the variational family of \cite{campbell2021online}. 
Moreover, \cite{mastrototaro:muller:olsson:2025} provided an online variational SMC framework that provably maximizes a well-defined asymptotic contrast function via stochastic approximation, which serves as the conceptual inspiration for our work.

\section{Model and background}
\label{sec:model:background}

Consider an SSM $(X_t, Y_t)_{t \in \nset}$, where $(X_t)_{t \in \nset}$ is a discrete-time Markov chain on $\Xsp \eqdef \statesp$. 
The distribution of $X_0$ has density $\chi$ with respect to the Lebesgue measure $\mu$ and for all $t \in \nset$, the conditional distribution of $X_{t+1}$ given $X_{0:t}$ depends only on $X_t$ and has transition density $\hd{t}(X_{t},\cdot)$. 
In SSMs, it is assumed that the states of the Markov chain are only partially observed through an observation process $(Y_t)_{t \in \nset}$ taking on values in $\Ysp \eqdef \obssp$. 
For every $t \in \nset$,  the observations $Y_{0:t}$ are assumed to be conditionally independent given $X_{0:t}$ and such that  the conditional distribution of each $Y_s$, $s \in \intset{0}{t}$, given $X_{0:t}$ depends only on the corresponding $X_t$ and has density $\md{t}(X_t,\cdot)$ with respect to the Lebesgue measure. 
The model is then defined, for all $ t \in \nset$, by the joint distributions 
\begin{equation}
\label{eq:joint:ssm:dist}
\jointd{0:t}(x_{0:t},y_{0:t}) \eqdef \chi(x_0) \md{0}(x_0, y_0)  \prod_{s=1}^{t} \hd{s}(x_{s-1}, x_s)\md{s}(x_s, y_s) 
\end{equation}
of the hidden states and the observations.

A classical learning task in SSMs is \textit{state inference}, which consists of estimating the joint-smoothing distribution, \emph{i.e.} the conditional distribution of $X_{0:t}$ given $Y_{0:t}$, given by
$$
\post{0:t}(x_{0:t}) \eqdef \frac{\jointd{0:t}(x_{0:t},y_{0:t})}{\jointd{0:t}(y_{0:t})} \eqsp,
$$
where $\jointd{0:t}(y_{0:t}) \eqdef \int \jointd{0:t}(x_{0:t},y_{0:t}) \, \rmd x_{0:t} $ is the observed-data likelihood. The marginal of this joint distribution with respect to the state $x_t$ at time $t$ is known as the \textit{filtering distribution} at time $t$, and its density with respect to the Lebesgue measure is denoted by $\post{t}$. It is straightforward to show that the density of the joint-smoothing distribution satisfies the so-called \textit{backward decomposition}
\begin{equation}
\label{eq:post:backward:factorization}
    \post{0:t}(x_{0:t}) = \post{t}(x_t)\prod_{s=1}^{t} \backwd{s-1|s}(x_s, x_{s-1}) \eqsp, 
\end{equation}
where each \emph{backward kernel}
\begin{equation}
    \label{eq:true_backwd}
    \backwd{s-1|s}(x_s, x_{s-1}) \propto 
    \hd{s}(x_{s - 1},x_{s})\post{s - 1}(x_{s-1})
\end{equation}
is the conditional probability density function of $X_{s-1}$ given $(X_{s}, Y_{0:s-1})$. The backward decomposition stems from the fact that the hidden process is still Markov when evolving conditionally on the observations, with time-inhomogeneous transition densities \eqref{eq:true_backwd}.  
However, since the filtering distributions are intractable, the backward kernels generally lack closed-form expressions. 

In variational approaches, the smoothing distribution $\post{0:t}$ is approximated by selecting a candidate from a parametric family $\{ \vd{0:t}\}_{\parvar \in \vparsp}$, known as the \textit{ variational family}, where $\vparsp$ is a parameter space. 
This parameter is learned jointly with $\thp$ by maximizing the ELBO 
\begin{equation}
\label{eq:ELBO}
\ELBO{t} = \pE_{\vd{0:t}}\left[\log \frac{\jointd{0:t}(X_{0:t},Y_{0:t})}{ \vd{0:t}(X_{0:t})}\right] \eqsp, 
\end{equation}
where $\pE_{\vd{0:t}}$ denotes expectation under $\vd{0:t}$. Needless to say, the form of the variational family is crucial in this approach. 
Motivated by the backward decomposition, some works impose structure on the variational family through a factorization of $\vd{0:t}$. 
A variational counterpart of \eqref{eq:post:backward:factorization}, introduced by \cite{campbell2021online}, is given by
\begin{equation}
\label{eq:varpost:backward:factorization}
 \vd{0:t}(x_{0:t})=  \vd{t}(x_t)\prod_{s=1}^{t}\vd{s-1\vert s}(x_{s},x_{s - 1})\eqsp,
\end{equation}
where $\vd{t}$ (resp. $\vd{s - 1\vert s}(x_{s},\cdot)$) are user-designed probability density functions, the variational kernels, whose parameters are learned from data.  A decisive advantage of this factorization is that it respects the true dependencies in \eqref{eq:post:backward:factorization}.
Moreover,  \cite{chagneux2022amortized} established an upper bound on the error when expectations with respect to the smoothing distribution are approximated by expectations with respect to variational 
distributions that satisfy this backward factorization.

\paragraph{Defining the variational distributions.}\label{sec:var_dist}
The variational inference framework requires the definition of the variational distributions involved in \eqref{eq:varpost:backward:factorization}, \textit{i.e.} the set of distributions $(\vd{t}, \vd{t| t+1})_{t\in \nset}$. We here define a new variational family in a recursive manner, with shared parameters over time. 
This new design based on backward factorization is efficient in terms of online parameter learning (as the number of parameters does not grow with $t$) and creates a link between variational kernels that will ensure an efficient importance sampling procedure in Section \ref{sec:MC:approach}. 

Our online learning challenge requires that variational distributions (i) can be recursively defined using streaming data $(y_t)_{t \in \nset}$, (ii) are related to each other to mimic the relation given by \eqref{eq:true_backwd} between the backward kernel and the filtering distribution, (iii) are easy to sample from to perform Monte Carlo approximations.
For this purpose, each $\vd{t}$ is chosen as a parametric distribution belonging to the exponential family (in our experiments, the Gaussian family in $\Xsp$), defined by some parameter $\eta_t$ belonging to a parameter space $\mathcal{E}$. More precisely, we define intermediate quantities $(a_t)_{t \in \nset}$  belonging to some user-defined space $\Asp$, initialized at some arbitrary value $a_0 \in \Asp$ and governed by a deterministic recursion $a_t = \amap{\parvar}(a_{t-1}, y_t)$. Based on these quantities, we let, for each $t$, $\eta_t = \fat(a_t)$. Here the mappings $\amap{\parvar}$ and $\fat$ are used-defined.
This framework creates a link between variational filtering distributions, in the spirit of the filtering recursions in SSMs.
The variational backward kernels are then defined from on the basis of this flow of distributions by setting, for all $t \in \nsetpos$, 
\begin{equation}
\label{eq:def:varbackwd}
\vd{t-1|t}(x_t, x_{t-1})  \propto \vd{t-1}(x_{t-1})\fwdpot{t}(x_{t-1}, x_t)\eqsp,
\end{equation}
where $(\fwdpot{t})_{t\geq 0}$ are potential functions on $\Xsp^2$ of form $\fwdpot{t}(x_{t-1}, x_t) = \exp( \langle \tilde{\eta}_t^{\parvar}(x_t), T(x_{t-1}) \rangle)$, with $\tilde{\eta}_{t}^\parvar(x_t) = \fxt(x_t)$ and $T(x_{t-1})$ being a natural parameter and a sufficient statistic, respectively, for the chosen exponential family.
Eqn.~\eqref{eq:def:varbackwd} ensures that $\vd{t-1|t}(x_t, \cdot)$ will be a probability density function with natural parameter $\eta_{t-1|t}^\parvar = \eta_{t-1}^\parvar + \tilde{\eta}_t^\parvar$. 
In this convenient setting, the backward kernels $\vd{t-1|t}$ can have arbitrarily complex dependencies on $x_t$, while their densities are derived analytically from the potentials.
This enables straightforward Monte Carlo sampling procedures and direct computations of normalizing constants (which are required in our proposed algorithm, \emph{e.g.}, in \eqref{eq:backwardweights:offline} below), while at the same time avoiding the reduction of our variational kernels to mere transformations or linearizations (\emph{e.g.}, linear Gaussian kernels). 
In is important to note that the parameters of functions $\amap{\parvar},\fat,\fxt$ are shared across time, leading to an amortized framework. 
In our experiments, these functions are neural networks, and $\parvar$ are their weights.

\section{Online variational learning}
\label{sec:onlinelearning}

\paragraph{The asymptotic contrast function and the COLBO.}
In the context of maximum likelihood estimation, the online learning of an unknown model parameter $\thp$ is known as recursive maximum likelihood (RML) \cite{legland1997recursive}.
RML focuses on maximizing the \emph{asymptotic contrast function} $\contr(\thp) \eqdef \lim_{t\to\infty} t^{-1} \log p_\thp(Y_{0:t})$ (a.s.), which serves as a foundational objective in this setting. 
If the data are generated by an SSM belonging to the parametric family of interest, characterized by a ‘true’ parameter $\thp^\ast$, then, under suitable identifiability conditions, the asymptotic contrast is maximised at $\thp^\ast$. 
Consequently, the maximum likelihood estimator (MLE) is strongly consistent in the sense that it converges almost surely to $\thp^\ast$ as $t$ tends to infinity.

Since the asymptotic contrast $\contr(\thp)$ is intractable, we use a similar approach to that of \cite{mastrototaro:muller:olsson:2025} and instead aim to maximise online, with respect to $(\thp, \parvar)$, the \emph{contrast lower bound} (COLBO) given by 
\begin{equation} \label{eq:def:COLBO}
\COLBO(\thp, \parvar) \eqdef \lim_{t\to\infty} \frac{1}{t}\ELBO{t}
\leq \contr(\thp) \quad \mbox{(a.s.)}.
\end{equation}

\paragraph{Stochastic approximation viewpoint. }

Following standard RML ideas, online variational learning seek to maximize $\COLBO(\thp, \parvar)$ by updating $(\thp, \parvar)$ in the direction of its gradient. 
Because $\COLBO(\thp, \parvar)$ is defined as a long-run time average, it is natural to pursue a stochastic approximation approach with the goal of solving 
\(
\nabla_{\thp, \parvar}\,\COLBO(\thp, \parvar)=0.
\)
Indeed, defining $\gdiffELBO{t}\eqdef \grad_{\thp,\parvar}\ELBO{t}-\grad_{\thp,\parvar}\ELBO{t-1}$, we may write  
\begin{equation} \label{eq:grad:elbo:finite}
\frac{1}{t}\grad_{\thp,\parvar}\,\ELBO{t} = \frac{1}{t}\sum_{s=1}^t\gdiffELBO{s} + \frac{1}{t} \grad_{\thp,\parvar} \ELBO{0}. 
\end{equation}
Interpreted through the lens of ergodic theory, it is tempting to see the long term limit of the right-hand side of \eqref{eq:grad:elbo:finite} as an expectation, allowing $\grad_{\thp,\parvar}\COLBO(\thp, \parvar)$ to be expressed as a \emph{mean field} (\emph{i.e.}, the deterministic drift) that governs the long-run behavior of Robbins--Monro stochastic updates. 
In this idealized framework, it is natural to use the observed gradient increments to build a sequence $(\thp_t,\parvar_t)_{t\in \nset}$  leading to an ideal procedure summarized in  Algorithm~\ref{alg:ideal}.
However, since each term $\gdiffELBO{s}$ depends on the whole historical record $Y_{0:s}$ and $(Y_t)_{t \in \nset}$ is not a Markov process, the existence of the limit of \eqref{eq:grad:elbo:finite} as $t$ tends to infinity is non-trivial.
Actually, to the best of our knowledge, no theoretical results exist justifying the existence of the COLBO objective \eqref{eq:def:COLBO} and its gradient, both of which are necessary to place this learning procedure on firm theoretical ground. In the coming sections we provide theoretical results motivating  this existence
for the variational family of \Cref{sec:model:background} (Eqn.~\eqref{eq:varpost:backward:factorization} and \eqref{eq:def:varbackwd}). 
In particular, we show that the COLBO objective and its gradient can be justified via the law of large numbers, applied to a suitably constructed Markov chain. On the basis of this justification, Algorithm 1 can be motivated as a stochastic approximation scheme with state-dependent Markov noise. 

\begin{algorithm}[b]
\caption{Ideal algorithm (exact recursions)}
\label{alg:ideal}
\begin{algorithmic}[1]
\STATE \textbf{For each} \(t \in \nset\), at the arrival of $y_t$, using Prop. \ref{prp:elbo:grad:recursion}:
\STATE \quad Compute \(\Hstat{t},\Gstat{t},\Fstat{t}\).
\STATE \quad Compute 
         $\gELBOt{\thp}{t}{\thp_t,\parvar_t}$ and $\gELBOt{\parvar}{t}{\thp_t,\parvar_t}$.
\STATE \quad \textbf{Update:} 
\begin{align*}
    \parvar_{t+1} 
           &\gets \parvar_t 
           + \gamma^{\parvar}_{t+1}\left(\gELBOt{\parvar}{t}{\thp_t,\parvar_t}
- \gELBOt{\parvar}{t-1}{\thp_t,\parvar_t}\right)\,,\\
    \thp_{t+1} 
           &\gets \thp_t 
           + \gamma^{\thp}_{t+1}\left(\gELBOt{\thp}{t}{\thp_t,\parvar_t}
- \gELBOt{\thp}{t-1}{\thp_t,\parvar_t}\right)\,,      
\end{align*}
        where $(\gamma^{\thp}_{t+1}, \gamma^{\parvar}_{t+1})$ are learning rates satisfying the usual Robbins Monro conditions.
\end{algorithmic}
\end{algorithm}

\paragraph{Recursive expression of the ELBO and its gradient. }
In the following, we assume that we are given a sequence $(y_t)_{t \in \nset}$ of observations, and leave the dependence on these implicit in the notation. 
Write $\addfelbo{0}(x_{-1}, x_0) \eqdef \log(\chi(x_0) \md{0}(x_0, y_0))$, $q^{\parvar}_{-1\vert 0}(x_0, x_{-1})=1$ 
and for $t \in \nsetpos$,
\begin{equation}
\label{eq:elbo:pair:terms}
\addfelbo{t}(x_{t-1}, x_{t}) \eqdef 
\log \left( \frac{\hd{t}(x_{t-1}, x_t)\md{t}(x_t, y_t)}{q^{\parvar}_{t-1\vert t}(x_t, x_{t-1})} \right),
\end{equation}
which allows to rewrite the ELBO \eqref{eq:ELBO} as 
 $$
 \ELBO{t} = \pE_{ \vd{0:t}}\left[\sum_{s = 0}^t \addfelbo{s}(x_{s-1}, x_{s}) - \log \vd{t}(x_t)\right]\,.
 $$
 The following proposition provides recursive formulas for computing both the ELBO and its gradient. 
 For brevity, we let $\mathbb{E}^{\parvar,x_t}_{t-1|t}$ denote expectation under $\vd{t-1|t}(x_t, \cdot)$.  
\begin{proposition}
\label{prp:elbo:grad:recursion}
For every $t \in \nset$ and $(\thp, \parvar) \in \thpsp \times \vparsp$, the ELBO and its gradient are given by 
\begin{align*}
\ELBO{t} &= \expect{\vd{t}}{\Hstat{t}(X_t)} - \expect{\vd{t}}{\log\vd{t}(X_t)} \eqsp, \\
\grad_\parvar \ELBO{t}&= \expect{\vd{t}}{\grad_\parvar \log \vd{t}(X_t) \, \Hstat{t}(X_t) + \Gstat{t}(X_t)}
\eqsp, \\
\grad_\thp \ELBO{t} &= \expect{\vd{t}}{\Fstat{t}(X_t)} \eqsp,
\end{align*}
where the real-valued function $\Hstat{t}$ on $\Xsp$ and its gradients $\Gstat{t} \eqdef \grad_{\parvar} \Hstat{t}$ and 
$\Fstat{t} \eqdef \grad_{\thp} \Hstat{t}$ satisfy the recursions
\begin{align*}
     \Hstat{t}(x_t) &= \mathbb{E}^{\parvar,x_t}_{t-1|t}{\left[\Hstat{t-1}(X_{t-1}) + \addfelbo{t}(X_{t-1}, x_t)\right]}\\
     \Gstat{t}(x_t) &= \mathbb{E}^{\parvar,x_t}_{t-1|t}{\left[\Gstat{t-1}\left(X_{t-1}\right) + \grad_\parvar \log \vd{t-1|t}(x_t, X_{t-1}) \{\Hstat{t-1}(x_{t-1}) + \addfelbo{t}(x_{t-1}, x_t)\} \right]}\\
     \Fstat{t}(x_t) &= \mathbb{E}^{\parvar,x_t}_{t-1|t} {\left[\Fstat{t-1}(X_{t-1}) + \grad_{\thp} \addfelbo{t}(X_{t-1}, x_t)\right]}\eqsp, 
\end{align*}
with $\Hstat{0}(x_0) = \addfelbo{0}(x_{-1}, x_{0}),\Gstat{0}(x_0) = 0,$ and $\Fstat{0}(x_0)=\grad_\thp\addfelbo{0}(x_{-1}, x_{0})$. 
\end{proposition}

\begin{proof}
See \Cref{appdx:proof:prp:elbo:recursion}. 
\end{proof}

Proposition~\ref{prp:elbo:grad:recursion} is of twofold interest. First, it provides a recursive scheme for the online computation of the ELBO, which will lead to a learning algorithm in Section~\ref{sec:MC:approach}. Second, it highlights the natural quantities to consider when studying the existence of the COLBO objective and its gradient.

\paragraph{Existence of the COLBO.} From now on, we assume that the observed data is generated by some SSM $(X_t, Y_t)_{t \in \nset}$, which does not necessarily belong to the parametric family considered in  \Cref{sec:model:background}. Under this assumption, it is easy to see that also the process $(Z_t)_{t \in \nset}$, where $Z_t \eqdef (X_t, Y_t, \Hstat{t}, \Gstat{t}, \Fstat{t}, a_t)$, with $\Hstat{t}$, $\Gstat{t}$, and $\Fstat{t}$ being the functions defined recursively in \Cref{prp:elbo:grad:recursion} and $a_t$ being the intermediate quantities used in the parameterization of $\vd{t}$, is a Markov chain. 
The state space and Markov kernel of $(Z_t)_{t \in \nset}$ are denoted by $(\Zsp, \Zfd)$ and $\emk$, respectively (see  \Cref{app:extended-chain}, Eqn.~\eqref{eq:def:T}, for details).
The Markov property follows from the assumed SSM dynamics of $(X_t, Y_t)_{t \in \nset}$, along with the fact that the updates of Proposition \ref{prp:elbo:grad:recursion}, as well as the update of $a_t$ from $a_{t - 1}$, are performed recursively based on the current observation $Y_t$. Denote also by $\S$ the Markov kernel of the marginal chain $(X_t, Y_t, a_t)_{t \in \nset}$. We will establish the exponential forgetting of the extended chain $(Z_t)_{t \in \nset}$ under the following assumptions.
\begin{assumption} \label{ass:unif:erg:xya:main}
    There exist $\pi \in \probmeas(\Xfd \tensprod \Yfd \tensprod \Afd)$ and $\alpha \in (0, 1)$ such that for every $t \in \nsetpos$, $\parvar \in \Phi$, and $(x, y, a) \in \Xsp \times \Ysp \times \Asp$,  
    $$
    \| (\S)^t(x, y, a) - \pi \|_{\mathsf{TV}} \leq \mixrt^t. 
    $$
\end{assumption}

\Cref{ass:unif:erg:xya:main} is discussed in \Cref{sec:discussion:unif:erg:xya}. The following assumptions are purely technical. 

\begin{assumption} \label{ass:grad:unif:bound:main}
    There exists $c \in \rsetpos$ such that for every $\parvar \in \Phi$, $a_t \in \Asp$, and $(x_{s + 1}, y_{s + 1}) \in \Xsp \times \Ysp$,     
    \begin{itemize}
        \item[(i)] $\mathbb{E}^{\parvar,x_{s+1}}_{s|s+1}\left[| \grad_\parvar \log \vd{s|s + 1}(x_{s + 1}, X_s) |^2 \right] \leq c^2$, 
        \item[(ii)] $\mathbb{E}^{\parvar,x_{s+1}}_{s|s+1}\left[|\addfelbo{s}(X_s, x_{s + 1})|^2 \,\right] \leq c^2$,  
        \item[(iii)] $\mathbb{E}^{\parvar,x_{s+1}}_{s|s+1}\left[ | \grad_\thp \addfelbo{s}(X_s, x_{s + 1}) | \right] \leq c$. 
    \end{itemize}
\end{assumption}

\begin{assumption}\label{ass:bounded-pot:main}
There exist constants $0<\varepsilon^{-}<\varepsilon^{+}<\infty$ such that, for every $t\in \nsetpos$, $(x_{t-1}, x_t) \in\X^2$, and $\parvar\in\Phi$,
\[
   \varepsilon^{-} \le \fwdpot{t}(x_{t-1},x_t) \le \varepsilon^{+}.
\]
\end{assumption}
The following theorem establishes the geometric ergodicity of the extended Markov chain, if not for all $f$ in the space $ \bmf(\Zfd)$ of bounded measurable functions on $\Zsp$, so at least for a subclass $\lip(\Zfd) \subset \bmf(\Zfd)$ of Lipschitz functions. More precisely, $f \in \lip(\Zfd)$ if there exists $\varphi \in \bmf(\Xfd \tensprod \Yfd \tensprod \Afd)$ such that for every $(x, y, a, h, h', u, u', v, v') \in \Xsp \times \Ysp \times \Asp \times \bmf(\Xfd)^6$, writing $\mathsf{s} = (x, y, a)$,
    \begin{itemize}
        \item[(i)] $|f(\mathsf{s}, h, u, v)| \leq \varphi(\mathsf{s})$, 
        \item[(ii)] $|f(\mathsf{s}, h, u, v) - f(\mathsf{s}, h', u', v')| \leq \varphi(\mathsf{s}) \left( \oscnorm(h - h') + \oscnorm(u - u') + \oscnorm(v - v') \right).$ 
    \end{itemize}

\begin{theorem}[Geometric ergodicity of $(Z_t)_{t \in \nset}$]
\label{th:main}
Assume  \ref{ass:unif:erg:xya:main}--\ref{ass:bounded-pot:main}. Then there exist $\rho \in (0, 1)$ and a functional $\statemap : \bmf(\Xfd)^6 \to \rsetpos$ such that for every $(\thp, \parvar) \in \Theta \times \vparsp$, $t \in \nset$, $f \in \lip(\Zfd)$, $z = (x_0, y_0, a_0, h_0, u_0, v_0) \in \Zsp$, and $z' = (x'_0, y'_0, a'_0, h'_0, u'_0, v'_0) \in \Zsp$,  
    \begin{equation}  \label{eq:K:diff:bound:final}
    |(\K)^t f(z) - (\K)^t f(z')|  \leq  \statemap(h_0, h'_0, u_0, u'_0, v_0, v'_0) \| \varphi \|_\infty \rho^t.     
    \end{equation}
    Moreover, there exists a functional $\bar{\statemap} : \bmf(\Xfd)^3 \to \rsetpos$ and a kernel $\Pi^{\thp, \parvar}$ such that for every $f \in \lip(\Zfd)$, $\Pi^{\thp, \parvar} f$ is constant and for every $t \in \nset$ and $z = (x_0, y_0, a_0, h_0, u_0, v_0) \in \Zsp$, 
    \begin{equation} \label{eq:erg:extended:chain}
    |(\K)^t f(z) - \Pi^{\thp, \parvar} f| \leq \bar{\statemap}(h_0, u_0, v_0) \| \varphi \|_\infty \rho^t.   
    \end{equation}
\end{theorem}
\begin{proof}
See \Cref{app:extended-chain}.
\end{proof}

Although the contraction \eqref{eq:erg:extended:chain} does not hold in the total variation norm (due to the restriction to test functions in $\lip(\Zfd)$), the quantity $\Pi^{\thp, \parvar}$ provided by the same theorem can be regarded as a candidate for the unique stationary distribution of $(Z_t)_{t \in \nset}$. Moreover, by \Cref{prp:elbo:grad:recursion}, each term $\gdiffELBO{s} = \gdiffELBO[Z_{s - 1:s}]{}$ depends explicitly on the consecutive states $Z_{s - 1:s}$ of the extended chain. By the law of large numbers for Markov chains, we may expect that (a.s.),  
\begin{equation} \label{eq:mean:field:representation}
    \lim_{t \to \infty} \frac{1}{t} \, \grad_{\thp,\parvar} \ELBO{t} 
    =
    \iint \gdiffELBO[z, z']{} \, 
    \Pi^{\thp, \parvar}(\rmd z) \, \K(z, \rmd z').   
\end{equation}
Letting the limit \eqref{eq:mean:field:representation} serve as the mean field of a stochastic approximation scheme with state-dependent Markov noise \cite{karimi:etal:2019}, a recursive Robbins--Monro algorithm finding a stationary point of the COLBO gradient is given by
$$
(\theta_{t + 1}, \parvar_{t + 1}) \gets (\theta_t, \parvar_t) + \gamma_{t + 1} \mathcal{G}^{\thp_t, \parvar_t} \langle Z_{t - 1}, Z_t \rangle,  
$$
and $Z_{t + 1} \sim T^{\thp_{t + 1}, \parvar_{t + 1}}(Z_t, \cdot)$, where  $(\gamma_t)_{t \in \nsetpos}$ is a sequence of step sizes satisfying the usual assumptions. This procedure is summarized in \Cref{alg:ideal}, which uses distinct step-size sequences $(\gamma_t^\thp)_{t \in \nsetpos}$ and $(\gamma_t^\parvar)_{t \in \nsetpos}$ for updating the model and variational parameters.

\section{Online Monte Carlo approximation}
\label{sec:MC:approach}

We now derive a practical version of the ideal Algorithm~\ref{alg:ideal}.  \Cref{prp:elbo:grad:recursion} suggests that it is possible to estimate the ELBO and its gradient recursively. The key feature of our Monte Carlo algorithm is that each conditional expectation in the recursion only needs to be estimated on a \textit{finite support}, bypassing the regression step required at each time step in  \cite{campbell2021online}. This results in a more efficient procedure, as confirmed empirically in Section \ref{sec:xp:chaotic:rnn}. The algorithm proceeds as follows:

First, sample $\lbrace \sample{0}{i} \rbrace_{i = 1}^N\overset{\text{i.i.d.}}{\sim} \vd{0}$, and set 
$$
\approxHstat{0}{i} = \Hstat{0}(\sample{0}{i}), \quad \approxGstat{0}{i} = \Gstat{0}(\sample{0}{i}), \quad \approxFstat{0}{i} = \Fstat{0}(\sample{0}{i})\,.
$$ 
At time $t \in \nsetpos$, having access to a Monte Carlo sample $\lbrace \sample{t - 1}{i} \rbrace_{1 = i}^N$ from $\vd{t - 1}$ and approximations $\approxHstat{t - 1}{i},\approxGstat{t - 1}{i},\approxFstat{t - 1}{i}$ of $\Hstat{t - 1}(\sample{t - 1}{i}),\Gstat{t - 1}(\sample{t - 1}{i}), \Fstat{t - 1}(\sample{t - 1}{i})$, respectively, 
sample independently $\lbrace \sample{t}{i} \rbrace_{i = 1}^N$ from $\vd{t}$ and update 
\begin{align}
\approxHstat{t}{i} =& \sum_{j=1}^N \nrmbackwdweight{t}{i,j}
\left(\approxHstat{t-1}{j} + \addfelbo{t}(\sample{t-1}{j}, \sample{t}{i}) \right), \label{eq:approx:Hstat:offline}\\ 
\approxGstat{t}{i} =& \sum_{i=1}^N \nrmbackwdweight{t}{i,j}\left\{\approxGstat{t - 1}{j} + \grad_\parvar \log \vd{t-1|t}(\xi_{t}^i, \xi_{t-1}^j)\left( \approxHstat{t-1}{j} + \addfelbo{t}(\xi_{t-1}^j, \xi_t^i) \right)\right\},  \label{eq:approx:Gstat:offline}\\
\approxFstat{t}{i} =& \sum_{i=1}^N \nrmbackwdweight{t}{i,j}\left(\approxFstat{t - 1}{j} + \grad_\thp \addfelbo{t}(\sample{t-1}{j}, \sample{t}{i})\right),  \label{eq:approx:Fstat:offline}
\end{align}
where
\begin{equation}
    \label{eq:backwardweights:offline}
    \nrmbackwdweight{t}{i,j} \eqdef \frac{\vd{t-1 \vert t}(\sample{t}{i}, \sample{t-1}{j})/\vd{t-1}(\sample{t-1}{j})}{\sum_{k=1}^N \vd{t-1 \vert t}(\sample{t}{i}, \sample{t-1}{k})/\vd{t-1}(\sample{t-1}{k})}\eqsp.
\end{equation}

Estimators (\ref{eq:approx:Hstat:offline}--\ref{eq:approx:Fstat:offline}) are \textit{self-normalized importance sampling} estimators  of the updates of Proposition \ref{prp:elbo:grad:recursion}, and \eqref{eq:backwardweights:offline} provides the (shared) importance weights of these estimators. 
Note that we cannot perform direct Monte Carlo approximation on the basis of samples from $\vd{t-1\vert t}(\sample{t}{i},\cdot)$, as we would not have access to any approximations of the values of the functionals $\Hstat{t-1}$ and $\Gstat{t-1}$ at the sampled points. 
The use of importance sampling is therefore a  prerequisite for updating the approximations. 
Moreover, note that the design of variational distributions imposed by \eqref{eq:def:varbackwd} creates a link between the the target $\vd{t - 1\vert t}$ distribution and $\vd{t}$, making the latter a natural proposal distribution.

Once $\{(\approxHstat{t}{i},\approxGstat{t}{i},\approxFstat{t}{i})\}_{i = 1}^N$ are computed, approximations of the ELBO and its gradient at time $t$ are obtained by 
\begin{align}
\approxELBO{t} &= \frac{1}{N}\sum_{i = 1}^N\left( \approxHstat{t}{i} - \log\vd{t}(\sample{t}{i})\right),\label{eq:approx:ELBO}\\
\hGELBO{\parvar}{t} &= \frac{1}{N}\sum_{i = 1}^{N}\left( \grad_\parvar \log \vd{t}(\sample{t}{i})  \approxHstat{t}{i} + \approxGstat{t}{i}\right), \label{eq:approx:grad:phi}\\
\hGELBO{\thp}{t} &= \frac{1}{N}\sum_{i = 1}^{N}\approxFstat{t}{i}. \label{eq:approx:grad:theta}
\end{align}
The full procedure, which we refer to as {\RMCVI} (Recursive Monte Carlo Variational Inference), is detailed in \Cref{alg:onlinenabla} in \Cref{appdx:alg}. 
Appendix \ref{appdx:full:algo} provides refinements to increase computational efficiency and reduce the variance of the gradient estimator.
It is worth noting that {\RMCVI} computes estimated gradients using both $(\parvar_t,\thp_t)$ and $(\parvar_{t-1},\thp_{t-1})$ (Eqns.~(\ref{eq:approx:update:phi}--\ref{eq:approx:update:theta}), Appendix~\ref{appdx:alg}),  which introduces a deviation from the ideal updates of Algorithm~\ref{alg:ideal}.
Such approximations---commonly employed in RML settings---are crucial for enabling a feasible  practical online implementation.

\section{Experiments}
\label{sec:exp}

We now evaluate the proposed algorithm on several streaming data inference tasks.
Our goal is to demonstrate that the method can jointly learn both the latent posterior approximation and the model parameters in a fully sequential manner while requiring substantially less computation than existing online approaches (Table \ref{table:1_step_smoothing_chaotic_rnn}).

In all experiments, the variational filtering distributions $\vd{t}$ are chosen to be in the Gaussian family. For the non-linear models  (Sections~\ref{sec:xp:chaotic:rnn} and \ref{sec:exp_air_quality})
we implement the deterministic recursion $a_t = \amap{\parvar}(a_{t-1}, y_t)$ as an RNN, where both the update function $\amap{\parvar}$ and the parameter mapping $\fat$ are parameterized by MLPs with tanh activation functions.
The variational backward kernels are defined via the potentials $\fwdpot{t}$ in \eqref{eq:def:varbackwd}, which are parameterized by similar neural networks (except for the linear Gaussian case, where we exploit analytical conjugation to derive exact backward kernels).
In addition to this architecture, some control-variate tricks are implemented to reduce the variance of the gradient estimator (see details in Appendix~\ref{appdx:experiments} and ~\ref{appdx:full:algo}).

\subsection{Linear-Gaussian HMM}
\label{sec:LGHMM}
We first assess our algorithm on a linear Gaussian SSM. Consider $(X_t, Y_t)_{t \in \nset}$ in $\mathbb{R}^{\statedim}\times \mathbb{R}^{\obsdim}$, with $X_0 \sim \mathcal{N}(\mu_0, Q_0)$, $Y_0 = G X_0 + \varepsilon_0$, and, for all $t \in \nsetpos$,
\begin{align*}
    X_t &= F X_{t-1} + \nu_t, & \nu_t \overset{\text{i.i.d}}{\sim} \mathcal{N}(0, I_{\statedim}),\\
    Y_t &= G X_t + \varepsilon_t,& \varepsilon_t \overset{\text{i.i.d}}{\sim} \mathcal{N}(0,  I_{\obsdim}),
\end{align*}
where $\theta = (F, G)$ is the parameter to be learned. In this case, the true smoothing distributions are Gaussian and can be computed via closed-form recursions (the Kalman smoother), providing an analytical reference for evaluation.
The variational family $\vd{0:t}$ is parameterized using Gaussian conditionals and marginals; see \Cref{sec:details:lin:Gauss} for details. 

The first experiment is run with $\statedim=\obsdim=10$, and a sequence of $T=50000$ observations. 
Figure \ref{fig:training_lgm_50k} shows the evolution of the ELBO through the learning of ($\thp,\parvar$) as well as the posterior mean of the hidden states on a test sequence of observations never seen by the model. 
The right panel shows a particular dimension of the hidden state for the test sequence, as well as the learned posterior distribution at different iterations.  As the number of observations grows, our estimator learns a mapping that produces the true posterior distribution, corresponding to the oracle Kalman smoother. 
Figure~\ref{fig:lgm_mae_compare} shows the MAE in parameter estimation of $F$ (a) and $G$ (b). Our method achieves lower error than the regression based method of \cite{campbell2021online} in roughly one quarter of the runtime when using $N=1000$ samples for the importance weights. 
Overall, these results highlight the computational efficiency of the proposed updates without compromising statistical accuracy.

\begin{figure}[t]
    \centering
    \includegraphics[width=\columnwidth]{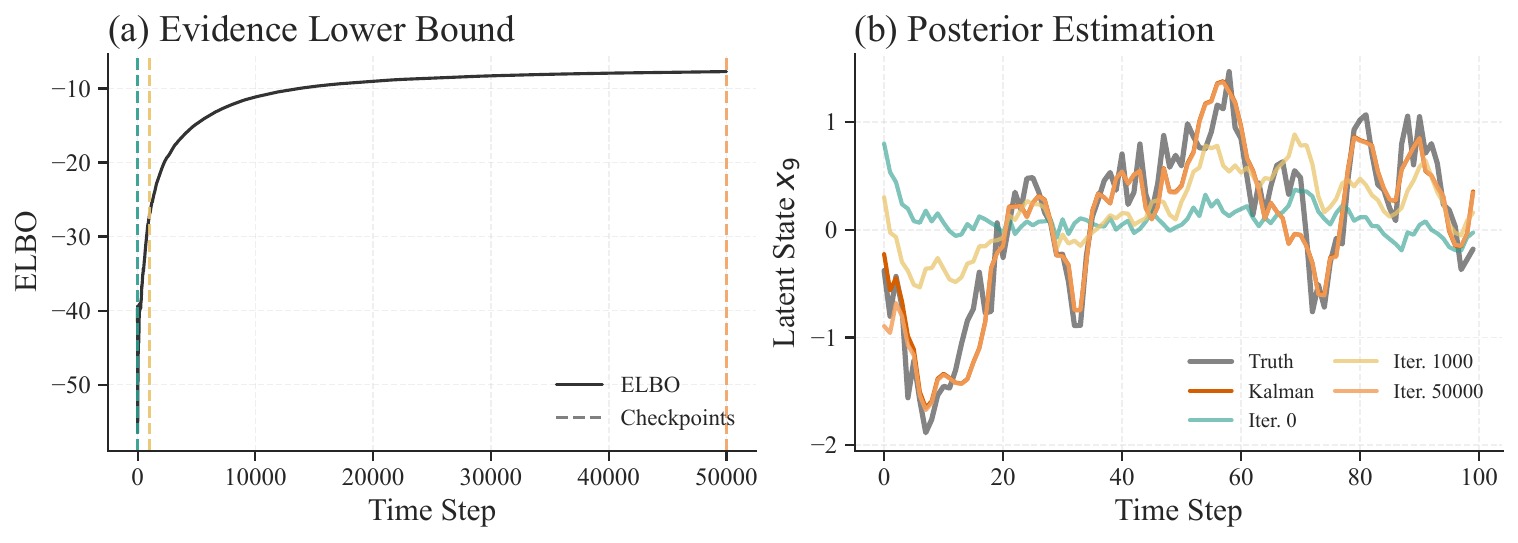}
    \caption{(a) Evolution of $\approxELBO{t}/t$ during the online learning in the linear Gaussian SSM. Vertical lines indicate times at which state estimation is made on a test sequence. (b) State estimation on a test sequence (only one particular dimension is displayed).}
    \label{fig:training_lgm_50k}
\end{figure}

\begin{figure}[t]
\centering
\includegraphics[width=\columnwidth]{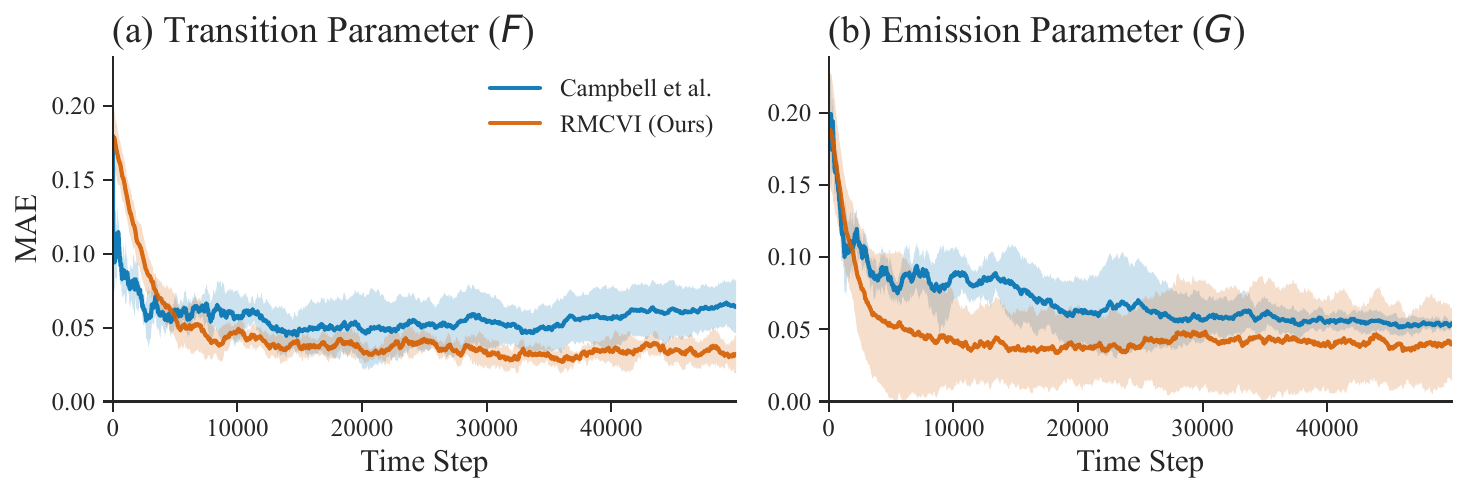}
\caption{Model-parameter learning in the linear–Gaussian HMM.
Mean-absolute errors of the transition ($F$) and emission ($G$) matrices for our method and \cite{campbell2021online}.}
\label{fig:lgm_mae_compare}
\end{figure}

\subsection{Chaotic recurrent neural network.}
\label{sec:xp:chaotic:rnn}
We now consider the model used in \cite{campbell2021online}, where $X_0 \sim \mathcal{N}(0, Q)$, $Y_0 = X_0 + \varepsilon_0$, and, for $t \in \nsetpos$,
\begin{align*}
X_t &= X_{t-1} + \frac{\Delta}{\rho}\left(\gamma W \tanh{(X_{t-1})} - X_{t-1}\right) + \eta_t\eqsp,\\
Y_t &= X_t + \varepsilon_t \eqsp,
\end{align*}
where $(\eta_t)_{t \in \nsetpos}$ and $(\varepsilon_t)_{t \in \nset}$ are mutually independent sequences of i.i.d. $\mathcal{N}(0, Q)$ and   Student-$t$  random variables, respectively.   
We set $\statedim=\obsdim=5$, and the use same true parameters as \cite{campbell2021online} (see Appendix~\ref{appdx:chaotic:rnn}) and $N=500$ importance samples.

\paragraph{Filtering and one-step smoothing. }
To compare our method with the one of \cite{campbell2021online} on this model, we reproduce the 1-step smoothing experiment of their work \cite[Appendix~B.2]{campbell2021online}.
Specifically, we evaluate the ability of both approaches to estimate the conditional laws of $X_{t-1}$ given $Y_{0:t}$ and of $X_t$ given $Y_{0:t}$ by learning $\expect{\vd{t-1:t}}{X_{t-1}}$ and $\expect{\vd{t}}{X_t}$.
In order to perform the same comparison, we mimic the non-amortized framework of the original paper (details are provided Appendix~\ref{appdx:chaotic:rnn}).
Table~\ref{table:1_step_smoothing_chaotic_rnn} reports the 1-step smoothing and filtering errors as defined in Eqn.~\eqref{eq:1-step_error}.
We also report average computation time per gradient steps. 
With comparable errors,  
{\RMCVI} is about $5$ times faster than the regression approach.

\begin{table}[b]
    \centering
    \small
    \begin{tabular}{||c||c|c|c||} 
        \hline
        Method & 1-Smooth. & Filt. & Time\\ [0.5ex] 
        \hline
        {\RMCVI} (ours) & 8.9 (0.2) & 10.3 (0.2) & 1 ms \\ 
        \hline
        \cite{campbell2021online} & 9.2 (0.2) & 10.3 (0.2) & 4.8 ms \\
        \hline
    \end{tabular}
    \caption{
    Time per gradient step, 1-step smoothing and filtering RMSE ($\times 10^{-2}$) (defined in Eqn.~\eqref{eq:1-step_error}) for the chaotic RNN.  
    }
    \label{table:1_step_smoothing_chaotic_rnn}
\end{table}

\paragraph{Online learning.}
Moving beyond the fixed-parameter setting, we evaluate our method in a true streaming regime where both the variational parameters and selected generative parameters are learned online. 
Concretely, $T=5\times 10^5$ observations are processed with updates of the parameters  $\gamma$ and  $\rho$.
Figure~\ref{fig:chaotic_learning} (a) shows the MAE between parameter estimates and the true value. 
The  central panel shows the dynamics of $\approxELBO{t}/t$ with checkpoint markers, at which state estimation on held-out sequences is performed (right panel) for a specific state dimension. 
These results show that the proposed scheme remains stable and accurate while simultaneously learning $(\rho,\gamma)$ in this highly nonlinear regime.

\begin{figure*}[ht]
    \centering
    \includegraphics[width=\textwidth]{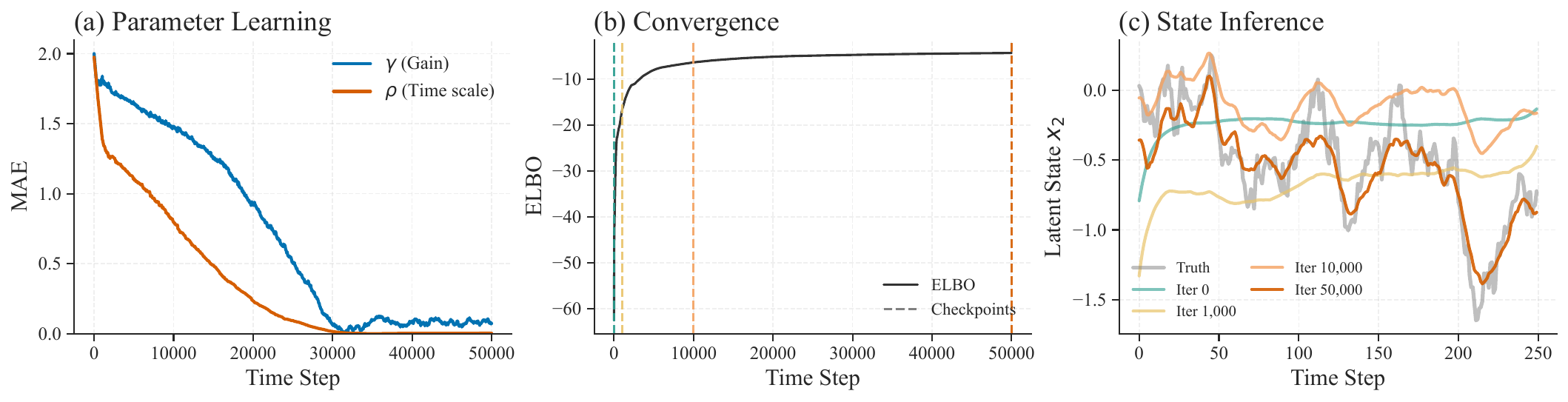}
    \caption{(a) Parameter MAE. (b) Approximate $\widetilde{\mathcal{L}}_t/t$ with checkpoint markers. (c) State estimation on a test sequence for one latent dimension. Colored lines/markers correspond to the same checkpoints.}
    \label{fig:chaotic_learning}
\end{figure*}

\subsection{Air-quality Data}
\label{sec:exp_air_quality}
We evaluate the framework on the UCI Air-Quality dataset \cite{uci_air_quality}, which consists of hourly averaged responses from a chemical sensor array alongside meteorological data. 
We process the data into an 8-dimensional observation vector $y_t \in \mathbb{R}^8$ spanning approximately one year ($T \approx 9,300$ steps). 
This benchmark is characterized by frequent periods of sensor failure, where valid signals are absent for extended durations. 
Rather than imputing these irregularities offline, we process the stream directly to rigorously test the method's ability to maintain coherent belief states during blackout periods. 
Visualizations of the data and more details on the signal characteristics and preprocessing are provided in Appendix~\ref{appdx:air_quality}.

We model the air quality dynamics using a non-linear Gaussian SSM with residual transitions. 
Let $X_t \in \mathbb{R}^{5}$ and $Y_t \in \mathbb{R}^{8}$ denote the latent state and observations respectively. The generative process is defined as:
\begin{align*}
    X_t &= X_{t-1} + f_\theta(X_{t-1}) + \nu_t \eqsp, & \nu_t \sim \mathcal{N}(0, Q_\theta) \\
    Y_t &= g_\theta(X_t) + \varepsilon_t \eqsp, & \varepsilon_t \sim \mathcal{N}(0, R_\theta)
\end{align*}
where $f_\theta$ and $g_\theta$ are neural networks parameterized by $\theta$ (with tanh activations), and the noise terms have diagonal covariance matrices, which are also learned. Here we used $N=20$ importance samples for the Monte Carlo estimates.

\paragraph{Online predictive performance.}
Our primary focus is the model's performance in an  online setting where all parameters must be learned from a cold start. 
We evaluate the model's ability to learn complex dynamics by measuring the one-step-ahead prediction RMSE on the five primary pollutants (CO, NO$_x$, NO$_2$, C$_6$H$_6$, O$_3$).
We compare our {\RMCVI} method against two baselines: a probabilistic online LSTM \cite{salinas2020deepar} with a Gaussian output  trained sequentially via maximum likelihood, and Online Variational SMC (OVSMC, \cite{mastrototaro:olsson:2024}).

Figure~\ref{fig:online_results} demonstrates the robustness and accuracy of our approach. As shown in Panel (a), we introduce a smoothing experiment with artificial sensor failure. While the purely autoregressive LSTM is limited to filtering and thus tracks the corrupted signal, {\RMCVI} effectively recovers the underlying ground truth. Similarly, panel (b) confirms that {\RMCVI} matches the predictive performance of the LSTM and outperforms OVSMC. 
Thus, {\RMCVI} combines the forecasting power of autoregressive networks with the advantage of variational smoothing, all while maintaining lower computational costs than particle-based methods.

\begin{figure}[ht]
    \centering
    \includegraphics[width=0.8\linewidth]{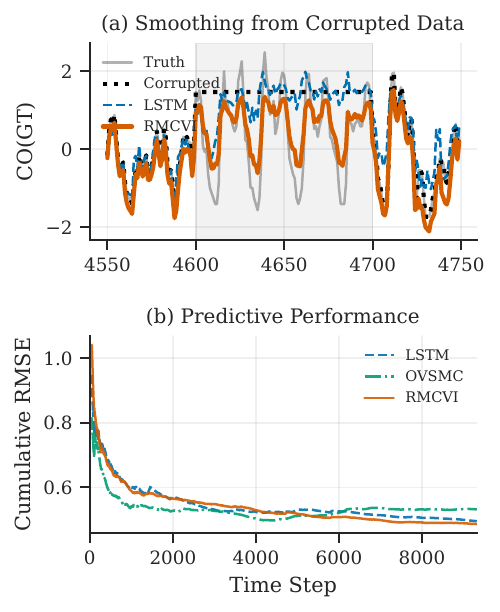}
    \vspace{-2mm}
    \caption{\textbf{Simultaneous Variational Learning and Prediction.} 
    (a) Smoothing reconstruction during sensor failure (black dotted). {\RMCVI} (orange) recovers the Truth (gray) while LSTM (blue) overfits the corruption.
    (b) Cumulative RMSE for one-step-ahead prediction averaged across all pollution features.}
    \label{fig:online_results}
\end{figure}

\section{Conclusion}

We introduced a theoretically grounded online variational learning algorithm for SSM. The performance of our method is assessed with synthetic and real-world datasets, where it is shown to be more efficient than recent alternatives both for smoothing and prediction tasks.
An important  future work concerns the theoretical analysis of the online Monte Carlo version of the algorithm, which is crucial to obtain quantitative guarantees and optimize hyperparameters.

\bibliography{online_vi}
\bibliographystyle{apalike}

\appendix

\section{Full online gradient estimator}
\label{appdx:alg}
\begin{algorithm}[H]
\caption{Online gradient estimator}\label{alg:onlinenabla}
\begin{algorithmic}
\STATE {\bfseries Input:} initial parameter estimate $(\thp_0, \parvar_0)$, step sizes $\{\gamma^\thp_t,\gamma^\parvar_t\}_{t \geqslant 1}$.
\STATE Sample $\lbrace \sample{0}{i} \rbrace_{i = 1}^N$ independently from $\vd[\parvar_0]{0}$\; 
\STATE For $1\leqslant i \leqslant n$, set $\approxHstat[\parvar_0]{0}{i} \gets \Hstat{0}(\sample{0}{i}),\approxGstat[\parvar_0]{0}{i} \gets 0, \approxFstat[\parvar_0]{0}{i} \gets \Fstat{0}(\sample{0}{i})$\;
\STATE set 
$\hGELBOt{\parvar}{0}{0} \gets N^{-1} \sum_{i = 1}^{N} \grad \log \vd[\parvar_0]{0}(\sample{0}{i}) \approxHstat[\parvar_0]{0}{i}$ and $\hGELBOt{\thp}{0}{0} \gets N^{-1} \sum_{i = 1}^{N}\approxFstat[\parvar_0]{0}{i}$\;
\STATE update $\parvar_1 \gets \parvar_0 + \gamma^{\parvar}_1\hGELBOt{\parvar}{0}{0}$ and $\thp_1 \gets \thp_0 + \gamma^{\thp}_1\hGELBOt{\thp}{0}{0}$\;
\FOR{$t \leftarrow 1$ to $n-1$}
\STATE sample $\lbrace \sample{t}{i} \rbrace_{i = 1}^N$ independently from $\vd[\parvar_t]{t}$\; 
\STATE set 
\begin{equation}
\label{eq:backwardweights:online}
    \nrmbackwdweight[\parvar_t]{t}{i,j} \gets \frac{\vd[\parvar_{t}]{t-1 \vert t}(\sample{t}{i}, \sample{t-1}{j})/\vd[\parvar_{t - 1}]{t-1}(\sample{t-1}{j})}{\sum_{k=1}^N \vd[\parvar_{t}]{t-1 \vert t}(\sample{t}{i}, \sample{t-1}{k})/\vd[\parvar_{t - 1}]{t-1}(\sample{t-1}{k})} \ ; 
\end{equation}
\STATE update 
\begin{align}
\approxHstat[\parvar_t]{t}{i} &\gets \sum_{j=1}^N \nrmbackwdweight[\parvar_t]{t}{i,j}
\left(\approxHstat[\parvar_{t-1}]{t-1}{j} + \addfelbo[\parvar_t]{t}(\sample{t-1}{j}, \sample{t}{i}) \right) \eqsp, \label{eq:approx:Hstat:online}\\ 
\approxGstat[\parvar_t]{t}{i} &\gets \sum_{i=1}^N \nrmbackwdweight[\parvar_t]{t}{i,j}\left\{\approxGstat[\parvar_{t-1}]{t - 1}{j} + \grad_\parvar \log \vd[\parvar_{t}]{t-1|t}(\xi_{t}^i, \xi_{t-1}^j)  \left( \approxHstat[\parvar_{t - 1}]{t-1}{j} + \addfelbo[\parvar_t]{t}(\xi_{t-1}^j, \xi_t^i) \right)\right\}\eqsp,
\label{eq:approx:Gstat:online}\\
\approxFstat[\thp_t]{t}{i} &\gets \sum_{i=1}^N \nrmbackwdweight[\parvar_t]{t}{i,j}\left\{\approxFstat[\thp_{t-1}]{t - 1}{j} + \grad_\thp \addfelbo[\parvar_t]{t}(\xi_{t-1}^j, \xi_t^i) \right\}\eqsp;
\label{eq:approx:Fstat:online}
\end{align}
\STATE set 
\begin{align*}
\hGELBOt{\parvar}{t}{t} &= \frac{1}{N}\sum_{i = 1}^{N}\left\{ \grad_\parvar \log \vd[\parvar_t]{t}(\sample{t}{i})  \approxHstat[\parvar_t]{t}{i} + \approxGstat[\parvar_t]{t}{i}\right\}\,, \\
\hGELBOt{\thp}{t}{t} &= \frac{1}{N}\sum_{i = 1}^{N}\approxFstat[\thp_t]{t}{i}\eqsp.
\end{align*}
\STATE Update
\begin{align}
\parvar_{t+1} \gets \parvar_t + \gamma^{\parvar}_{t+1}\left(\hGELBOt{\parvar}{t}{t} - \hGELBOt{\parvar}{t-1}{t-1}\right)\label{eq:approx:update:phi}\,\\
\thp_{t+1} \gets \thp_t + \gamma^{\thp}_{t+1}\left(\hGELBOt{\thp}{t}{t} - \hGELBOt{\thp}{t-1}{t-1}\right)\label{eq:approx:update:theta}\,;
\end{align}
\ENDFOR
\end{algorithmic}
\end{algorithm}

\section{Proof of proposition \ref{prp:elbo:grad:recursion}}
\label{appdx:proof:prp:elbo:recursion}

By definition,
$$
\ELBO{t} = \expect{\vd{0:t}}{\log \frac{\jointd{0:t}(X_{0:t},Y_{0:t})}{ \vd{0:t}(X_{0:t})}} =  \expect{\vd{0:t}}{\log \frac{\chi(X_0) \md{0}(X_0, Y_0)  \prod_{s=1}^{t} \hd{s}(X_{s-1}, X_s)\md{s}(X_s, Y_s)}{\vd{t}(X_t)\prod_{s = 1}^t\vd{s-1\vert s}(X_{s}, X_{s - 1})}}\eqsp,
$$
which, using \eqref{eq:elbo:pair:terms}, yields 
$$
\ELBO{t} =  \expect{\vd{0:t}}{\sum_{s = 0}^t \addfelbo{s}(X_{s-1}, X_s) - \log\vd{t}(X_t)}\eqsp.
$$
Therefore, 
\begin{equation}
\ELBO{t} =\expect{\vd{t}}{\Hstat{t}(X_t)} - \expect{\vd{t}}{\log\vd{t}(X_t)}\eqsp,\label{eq:ELBO:wrt:h}
\end{equation}
where 
\begin{align*}
\Hstat{t}(x_t) = 
\int \left(\sum_{s = 0}^t \addfelbo{s}(x_{s-1}, x_s)\right) \prod_{s = 1}^{t}\vd{s-1\vert s}(x_{s}, x_{s-1}) \, \rmd x_{0:t-1} = \expect{\vd{0:t-1\vert t}(x_t,\cdot)}{\afelbo{0:t}(X_{0:t-1}, x_t)}\eqsp,
\end{align*}
with $\vd{0:t-1\vert t}(x_t, x_{0:t-1}) = \prod_{s = 1}^t  \vd{s-1\vert s}(x_{s}, x_{s-1})$, and we introduced the notation $\afelbo{0:{t}}(x_{0:t}) = \sum_{s = 0}^{t}\addfelbo{s}(x_{s-1}, x_{s})$. 
Then, note that
\begin{align*}
\Hstat{t}(x_t) 
 &= \int \left(\Hstat{t -1}(x_{t-1}) + \addfelbo{t}(x_{t-1}, x_t)\right)\vd{t-1\vert t}(x_{t}, x_{t-1}) \, \rmd x_{t-1}\\
 &= \expect{\vd{t-1\vert t}(x_t,\cdot)}{\Hstat{t -1}(X_{t-1}) + \addfelbo{t}(X_{t-1}, x_t)}\eqsp.
\end{align*}
This establishes the recursive expression of the ELBO. 

Now, we consider the gradient of the ELBO with respect to $\parvar$. Write 
\begin{align*}
    \grad_{\parvar} \ELBO{t} &= \grad_{\parvar}  \expect{\vd{t}}{\Hstat{t}(X_t)} - \grad_{\parvar}  \expect{\vd{t}}{\log\vd{t}(X_t)}\eqsp\\
    &= \grad_{\parvar} \int \left(\Hstat{t}(x_t) - \log \vd{t}(x_t)\right)\vd{t}(x_t) \, \rmd x_t\\
    &= \int \left(\grad_{\parvar} \Hstat{t}(x_t) - \grad_{\parvar} \log \vd{t}(x_t)\right)\vd{t}(x_t) \, \rmd x_t + \int \left(\Hstat{t}(x_t) - \log \vd{t}(x_t)\right) \grad_{\parvar}  \vd{t}(x_t) \, \rmd x_t \\
   &= \expect{\vd{t}}{\grad_{\parvar} \Hstat{t}(X_t) + \left(\Hstat{t}(X_t) - \log \vd{t}(X_t)\right)  \grad_{\parvar} \log \vd{t}(X_t)}\eqsp, 
\end{align*}
where we used that $\mathbb{E}_{\vd{t}}[\grad\log\vd{t}(X_t)] = 0$. Then, writing $\Gstat{t}(x_t) = \grad_{\parvar} \Hstat{t}(x_t)$,
\begin{align*}
    \Gstat{t}(x_t) &= \grad_{\parvar} \expect{\vd{0:t-1\vert t}(x_t,\cdot)}{\afelbo{0:t}(X_{0:t-1}, x_t)}\\
    &=\expect{\vd{0:t-1\vert t}(x_t,\cdot)}{\left(\grad_{\parvar} \log \vd{0:t-1\vert t} \times \afelbo{0:t}\right)(X_{0:t-1}, x_t)}+ \expect{\vd{0:t-1\vert t}(x_t,\cdot)}{\grad_{\parvar} \afelbo{0:t}(X_{0:t-1}, x_t)}\eqsp.
\end{align*}
Remembering that 
$$
\afelbo{0:t} = \log\left(\chi(x_0) \md{0}(x_0, y_0)\right) + \sum_{s=1}^t \log \left(\hd{s}(x_{s-1}, x_s)\md{s}(x_s, y_s)\right) - \sum_{s=1}^t \log \vd{s-1\vert s}(x_s, x_{s-1}),
$$
we obtain that $\grad_{\parvar} \afelbo{0:t}(X_{0:t-1}, x_t) = -\grad_{\parvar} \log \vd{0:(t-1)\vert t}(X_{0:t-1}, x_t)$, 
which has zero expectation under $\vd{0:t-1\vert t}(x_t,\cdot)$. Thus, 
$$
\Gstat{t}(x_t) = \expect{\vd{0:t-1\vert t}(x_t,\cdot)}{\left(\grad_{\parvar} \log \vd{0:(t-1)\vert t} \times \afelbo{0:t}\right)(X_{0:t-1}, x_t)}\eqsp.
$$

To establish the recursion for $\Gstat{t}(x_t)$, write 
\begin{align}
    \Gstat{t}(x_t) &= \mathbb{E}_{\vd{0:t-1\vert t}(x_t,\cdot)}\left[\left(\grad_{\parvar} \log \vd{0:(t-2)\vert t-1}(X_{0:t-1}) + \grad_{\parvar} \log \vd{t - 1\vert t}(X_{t-1}, x_{t}) \right) \right. \nonumber \\
    &\hspace{5.5cm}\left.\times\left(\afelbo{0:t-1}(X_{0:t-1}) + \addfelbo{t}(X_{t-1}, x_t)\right)\right] \label{eq:appdx:beg:G}\\
    &= \expect{\vd{t-1\vert t}(x_t,\cdot)}{\Gstat{t-1}(X_{t-1})} \label{eq:appdx:Gtm1} \\
    &\ +  \expect{\vd{t-1\vert t}(x_t,\cdot)}{\grad_{\parvar} \log \vd{t - 1\vert t}(X_{t-1}, x_{t})\left(\expect{\vd{0:(t-2)\vert t-1}}{\afelbo{0:t-1}(X_{0:t-1})} + \addfelbo{t}(X_{t-1}, x_t)\right)} \nonumber\\
    &\ +\expect{\vd{t-1\vert t}(x_t,\cdot)}{\addfelbo{t}(X_{t-1}, x_t)\times \expect{\vd{0:(t-2)\vert t-1}}{\grad_{\parvar} \log \vd{0:(t-2)\vert t-1}(X_{0:t-1})}}\eqsp,  \nonumber 
\end{align}
which yields
$$
\Gstat{t}(x_t) = \expect{\vd{t-1\vert t}}{\Gstat{t-1}(X_{t-1}) + \grad_{\parvar} \log \vd{t - 1\vert t}(X_{t-1}, x_{t})\times \left(\Hstat{t-1}(X_{t-1}) + \addfelbo{t}(X_{t-1}, x_t) \right)}\eqsp, 
$$
which was to be established. 

Finally, let us consider the gradient w.r.t. $\thp$. 
Using \eqref{eq:ELBO:wrt:h}, we have that $\grad_\thp\ELBO{t} = \expect{\vd{t}}{\grad_\thp\Hstat{t}(X_t)}$. 
Writing $\Fstat{t} = \grad_\thp\Hstat{t}$, we obtain
\begin{align*}
    \Fstat{t}(x_t) &= \grad_\thp\expect{\vd{0:t-1\vert t}(x_t,\cdot)}{\afelbo{0:t}(X_{0:t-1}, x_t)}\\
    &= \expect{\vd{0:t-1\vert t}(x_t,\cdot)}{\grad_\thp\afelbo{0:t-1}(X_{0:t-2}, X_{t-1}) + \grad_\thp\afelbo{t-1}(X_{t-1}, x_{t})}\\
    &= \expect{\vd{0:t-1\vert t}(x_t,\cdot)}{\grad_\thp\afelbo{0:t-1}(X_{0:t-2}, X_{t-1})} + \expect{\vd{t-1\vert t}(x_t,\cdot)}{\grad_\thp\afelbo{t-1}(X_{t-1}, x_{t})}\\
    &=\expect{\vd{t-1\vert t}(x_t,\cdot)}{\Fstat{t-1}(X_{t-1})} + \expect{\vd{t-1\vert t}(x_t,\cdot)}{\grad_\thp\afelbo{t-1}(X_{t-1}, x_{t})}\,,
\end{align*}
which concludes the proof.

\section{Proof of Theorem~\ref{th:main}}
\label{app:extended-chain}


\subsection{Geometric ergodicity of the extended chain}

First, some notation. Let  $(\mathsf{E}, \mathcal{E})$ be an arbitrary state space. Then a kernel $K$ on $\mathsf{E} \times \mathcal{E}$ induces two endomorphisms, the first acting on the space $\bmf(\mathcal{E})$ of bounded measurable functions on $(\mathsf{E}, \mathcal{E})$ according to 
$$
\bmf(\mathcal{E}) \ni f \mapsto Kf(\cdot) \eqdef \int f(x) \, K(\cdot, \rmd x) \in \bmf(\mathcal{E})  
$$
and the second acting on the space $\mathsf{M}(\mathcal{E})$ of measures on $(\mathsf{E}, \mathcal{E})$ (we let $\probmeas(\mathcal{E}) \subset \mathsf{M}(\mathcal{E})$ denote the subspace of probability measures) according to 
$$
\mathsf{M}(\mathcal{E}) \ni \mu \mapsto \mu K(x, \cdot) \eqdef \int \mu(\rmd x) \, K(x, \cdot) \in \mathsf{M}(\mathcal{E}).   
$$
In addition, the product of two kernels $K$ and $L$ on $(\mathsf{E}, \mathcal{E})$ is defined as the kernel 
$$
KL : \mathsf{E} \times \mathcal{E} \ni (x, A) \mapsto \int K(x, \rmd x') \, L(x', A)
$$
on $(\mathsf{E}, \mathcal{E})$. Using this notation, we may define, for any $t \in \nsetpos$, the power $K^t$ of a kernel $K$ by multiplying $K$ by itself $t - 1$ times. 

In the following we denote, for every $t \in \nset$, by 
$$
\vk{t|t + 1}f(x_{t + 1}) \eqdef \int f(x_t) \vd{t|t + 1}(x_t, x_{t + 1}) \, \rmd x_t, \quad (x_{t + 1}, f) \times \bmf(\Xfd), 
$$
the Markov kernel induced by the transition density $\vd{t|t + 1}$.  

In \Cref{sec:onlinelearning} it is assumed that the data generating process $(X_t, Y_t)_{t \in \nset}$ is an SSM, and we denote by $\R$ its Markov transition kernel of this process. Since the auxiliary states $(a_t)_{t \in \nset}$ are generated deterministically from the observations via the mapping $\amap{\parvar}$, also the augmented process $(X_t, Y_t, a_t)_{t \in \nset}$ is Markov with transition kernel
\begin{multline*}
\S f(x_t, y_t, a_t) \eqdef \int f(x_{t + 1}, y_{t + 1}, \amap{\parvar}(a_t, y_{t + 1})) \, \R(x_t, y_t, \rmd (x_{t + 1}, y_{t + 1})), \\ (x_t, y_t, a_t, f) \in \Xsp \times \Ysp \times \Asp \times \bmf(\Xfd \tensprod \Yfd \tensprod \Afd). 
\end{multline*}

For clarity, we now briefly recall the principal assumptions of this work.

\begin{assumption}[Uniform ergodicity of $(X_t, Y_t, a_t)_{t \in \nset}$] \label{ass:unif:erg:xya}
    There exist $\pi \in \probmeas(\Xfd \tensprod \Yfd \tensprod \Afd)$ and $\alpha \in (0, 1)$ such that for every $t \in \nsetpos$, $\parvar \in \Phi$, and $(x, y, a) \in \Xsp \times \Ysp \times \Asp$,  
    $$
    \| (\S)^t(x, y, a) - \pi \|_{\mathsf{TV}} \leq \mixrt^t. 
    $$
\end{assumption}

\Cref{ass:unif:erg:xya} is discussed in Section~\ref{sec:discussion:unif:erg:xya} below.

\begin{assumption}\label{ass:bounded-pot}
There exist constants $0 < \varepsilon^- < \varepsilon^+ < \infty$ such that for every $\parvar \in \Phi$, $t \in \nsetpos$, and 
$(x, x') \in \X^2$,
$$
   \varepsilon^{-} \le \fwdpot{t}(x, x') \le \varepsilon^+.
$$
In addition, we let $\varepsilon \eqdef \varepsilon^- / \varepsilon^+$. 
\end{assumption}

\begin{assumption} \label{ass:grad:unif:bound}
    There exists $c \in \rsetpos$ such that for every $\parvar \in \Phi$, $a_t \in \Asp$, and $(x_{s + 1}, y_{s + 1}) \in \Xsp^2 \times \Ysp$, 
    \begin{itemize}
        \item[(i)] $\displaystyle
        \int | \grad_\parvar \log \vd{s|s + 1}(x_{s + 1}, x_s) |^2 \, \vk{s|s + 1}(x_{s + 1}, \rmd x_s) \leq c^2$, 
        \item[(ii)] $\displaystyle
        \int |\addfelbo{s}(x_s, x_{s + 1})|^2 \, \vk{s|s + 1}(x_{s + 1}, \rmd x_s) \leq c^2$,  
        \item[(iii)] $\displaystyle
        \int | \grad_\thp \addfelbo{s}(x_s, x_{s + 1}) | \, \vk{s|s + 1}(x_{s + 1}, \rmd x_s) \leq c$. 
    \end{itemize}
\end{assumption}

From now on we let $\Z \eqdef \Xsp \times \Ysp \times \Asp \times \bmf(\Xfd)^3$ denote the state space of the extended chain $(Z_t)_{t \in \nset}$ and let $\Zfd \eqdef \Xfd \tensprod \Yfd \tensprod \Afd \tensprod \bmffd(\Xfd)^{\tensprod 3}$ be the associated $\sigma$-field.

\begin{definition}
\label{def:lip:h}
    Let $\lip(\Zfd)$ be the set of $f \in \mf(\Zfd)$ for which there exists $\varphi \in \bmf(\Xfd \tensprod \Yfd \tensprod \Afd)$ such that for every $(x, y, a, h, h', u, u', v, v') \in \Xsp \times \Ysp \times \Asp \times \bmf(\Xfd)^6$, 
    \begin{itemize}
        \item[(i)] $|f(x, y, a, h, u, v)| \leq \varphi(x, y, a)$, 
        \item[(ii)] $|f(x, y, a, h, u, v) - f(x, y, a, h', u', v')| \leq \varphi(x, y, a) \left( \oscnorm(h - h') + \oscnorm(u - u') + \oscnorm(v - v') \right) $.  
    \end{itemize}
\end{definition}

The following is a slightly more precise statement of \Cref{th:main}, our main result. 

\begin{theorem} \label{prop:K_2:ergodicity}
    Assume \ref{ass:unif:erg:xya}, \ref{ass:grad:unif:bound}, and \ref{ass:bounded-pot}. Then there exist $\rho \in (0, 1)$ and a function $\statemap : \bmf(\Xfd)^6 \to \rsetpos$ such that for every $(\thp, \parvar) \in \Theta \times \vparsp$, $t \in \nset$, $f \in \lip(\Zfd)$, $z = (x_0, y_0, a_0, h_0, u_0, v_0) \in \Zsp$, and $z' = (x'_0, y'_0, a'_0, h'_0, u'_0, v'_0) \in \Zsp$,  
    \begin{equation}  \label{eq:K:diff:bound:final:th}
    |(\K)^t f(z) - (\K)^t f(z')| \leq  \statemap(h_0, h'_0, u_0, u'_0, v_0, v'_0) \| \varphi \|_\infty \rho^t.     
    \end{equation}
    Moreover, there exists a kernel $\Pi^{\thp, \parvar}$ on $\Zsp \times \lip(\Zfd)$ such that for every $f \in \lip(\Zfd)$, $\Pi^{\thp, \parvar} f$ is constant and for every $t \in \nset$ and $z = (x_0, y_0, a_0, h_0, u_0, v_0) \in \Zsp$, 
    \begin{equation} \label{eq:K:Sigma:diff}
    |(\K)^t f(z) - \Pi^{\thp, \parvar} f| \leq \bar{\statemap}(h_0, u_0, v_0) \| \varphi \|_\infty \rho^t,  
    \end{equation}
    where 
    \begin{equation} \label{eq:def:statemap:bar}
        \bar{\statemap}(h_0, u_0, v_0) \eqdef \frac{1}{1 - \rho} \int \statemap(h_0, h'_0, u_0, u'_0, v_0, v'_0) \, \K(z, \rmd z'). 
    \end{equation}
\end{theorem}

We preface the proof of \cref{prop:K_2:ergodicity} with a couple of definitions and lemmas. First, we summarize the recursion for the ELBO using the function-valued mapping 
$$
\hmap{\thp, \parvar}(\cdot, a_t, y_{t + 1}) : h_t \mapsto \int \left( h_t(x_t) + \addfelbo{t}(x_t, x_{t + 1}) \right) \vk{t|t + 1}(x_{t + 1}, \rmd x_t),  
$$
where the kernel $\vk{t|t + 1}$ and the term $\addfelbo{t}$ depend implicitly on $a_t$ and $y_{t + 1}$, respectively, implying that $h_{t + 1}(x_{t + 1}) = \hmap{\thp, \parvar}(h_t, a_t, y_{t + 1})(x_{t + 1})$. 
Based on the latter, we also define, for every $t \in \nsetpos$ and  vector $y_{1:t} \in \Ysp^t$, the composite versions 
\begin{equation} %
\label{eq:def:hmap:comp}
\hmap{\thp,\parvar}(\cdot, a_{0:t - 1}, y_{1:t}) : h_0 \mapsto  
\begin{cases}
     \hmap{\thp,\parvar}(h_0, a_0, y_1), & \mbox{for $t = 1$,} \\
     \hmap{\thp,\parvar}(\hmap[t - 1]{\thp,\parvar}(h_0, a_{0:t - 2}, y_{1:t - 1}), a_{t - 1}, y_t), & \mbox{for $t \geq 2$.} \end{cases}
\end{equation}

Using the similar notation 
$$
\vmap{\thp, \parvar}(\cdot, v_t, a_t) : v_t \mapsto \int \left( v_t(x_t) + \grad_\thp \addfelbo{t}(x_t, x_{t + 1}) \right) \vk{t|t + 1}(x_{t + 1}, \rmd x_t), 
$$
for the recursion of the ELBO gradient with respect to $\thp$, the composite mappings $(\vmap[t]{\thp, \parvar})_{t \in \nsetpos}$ are defined similarly. 

Finally, letting  
$$
\umap{\thp, \parvar}(\cdot, a_t, y_{t + 1}) : (h_t, u_t) \mapsto \int \left( u_t(x_t) + \grad_\parvar \log \vd{t|t + 1}(x_t, x_{t + 1}) \{ h_t(x_t) +
\addfelbo{t}(x_t, x_{t + 1}) \} \right) \vk{t|t + 1}(x_{t + 1}, \rmd x_t)
$$
summarize the recursion for the ELBO gradient with respect to $\parvar$, so that $u(x_{t + 1}) = \umap{\thp, \parvar}(h_t, u_t, a_t, y_{t + 1})(x_{t + 1})$, we also define the compositions 
\begin{equation}
\label{eq:def:umap:comp}
\umap[t]{\thp,\parvar}(\cdot, a_{0:t - 1}, y_{1:t}) : (h_0, u_0) \mapsto  
\begin{cases}
     \umap{\thp,\parvar}(h_0, u_0, a_0, y_1), & \mbox{for $t = 1$,} \\
     \umap{\thp,\parvar}(\hmap[t - 1]{\thp,\parvar}(h_0, a_{0:t - 2}, y_{1:t - 1}), \umap[t - 1]{\thp,\parvar}(h_0, u_0, a_{0:t - 2}, y_{1:t - 1}), a_{t - 1}, y_t), & \mbox{for $t \geq 2$.}
\end{cases}
\end{equation}

Using these definitions, the transition kernel $\K$ of the extended chain $(Z_t)_{t \in \nset}$ can be expressed as 
\begin{multline} \label{eq:def:T}
    \K f (z_t) = \int f(x_{t + 1}, y_{t + 1}, \amap{\thp, \parvar}(a_t, y_{t + 1}), \hmap{\thp, \parvar}(h_t, a_t, y_{t + 1}), \umap{\thp, \parvar}(h_t, u_t, a_t, y_{t + 1}), \vmap{\thp, \parvar}(v_t, a_t, y_{t + 1})) \\ 
    \times \R((x_t, y_t), \rmd (x_{t + 1}, y_{t + 1})), \quad (z_t, f) \in \Zsp \times \bmf(\Zfd),   
\end{multline}
where $z_t = (x_t, y_t, a_t, h_t, u_t, v_t)$.

The following lemma establishes the geometric contraction of the function-valued mappings defined above. 

\begin{lemma}
\label{hmap:vmap:contraction:lemma}
Assume \ref{ass:bounded-pot}. Then for every $t \in \nset$, $y_{1:t} \in \Ysp^t$, $a_{0:t - 1} \in \Asp^t$, and $(h_0, h'_0, v_0, v'_0) \in \bmf(\Xfd)^4$, 
\begin{itemize}
\item[(i)]
    $\displaystyle \oscnorm \left( \hmap[t]{\thp, \parvar}(h_0, a_{0:t - 1}, y_{1:t}) - \hmap[t]{\thp, \parvar}(h_0', a_{0:t - 1}, y_{1:t}) \right) \leq (1 - \varepsilon)^t \oscnorm(h_0 - h'_0)$, 
    \item[(ii)]
    $\displaystyle \oscnorm \left( \vmap[t]{\thp, \parvar}(v_0, a_{0:t - 1}, y_{1:t}) - \vmap[t]{\thp, \parvar}(v_0', a_{0:t - 1}, y_{1:t}) \right) \leq (1 - \varepsilon)^t \oscnorm(v_0 - v'_0)$, 
    \end{itemize}
    where $\varepsilon \in (0, 1)$ is given in \Cref{ass:bounded-pot}. 
    \begin{itemize}
    \item[(iii)] Assume additionally \ref{ass:grad:unif:bound}\,(i). Then for every $\varrho \in (1 - \varepsilon, 1)$ there exists $d > 0$ such that for every $t \in \nset$, $y_{1:t} \in \Ysp^t$, $a_{0:t - 1} \in \Asp^t$, and $(h_0, h'_0, u_0, u'_0) \in \bmf(\Xfd)^4$, 
    $$ 
    \oscnorm \left( \umap[t]{\thp,\parvar}(h_0, u_0, a_{0:t - 1}, y_{1:t}) - \umap[t]{\thp,\parvar}(h'_0, u'_0, a_{0:t - 1}, y_{1:t}) \right) \leq (1 - \varepsilon)^t \oscnorm(u_0 - u'_0) + d \varrho^t \oscnorm(h_0 - h'_0). 
    $$ 
\end{itemize}
\end{lemma}

The proof of \Cref{hmap:vmap:contraction:lemma} is based on the following lemmas. 

\begin{lemma} \label{lem:ump:explicit:form}
    for every $t \in \nsetpos$, $y_{1:t} \in \Ysp^t$, $a_{0:t - 1} \in \Asp^t$,  $(h_0, u_0) \in \bmf(\Xfd)^2$, and $x_t \in \Xsp$, 
    \begin{multline*}
    \umap[t]{\thp,\parvar}(h_0, u_0, a_{0:t - 1}, y_{1:t})(x_t) \\ 
    = \idotsint \left( u_0(x_0) + \sum_{s = 0}^{t - 1} \grad_\parvar \log \vd{s|s + 1}(x_{s + 1}, x_s) \{ \hmap[s]{\thp, \parvar}(h_0, a_{0:s - 1}, y_{1:s}) + \addfelbo{s}(x_s, x_{s + 1}) \} \right) \prod_{s = 0}^{t - 1} \vk{s|s + 1}(x_{s + 1}, \rmd x_s) 
    \end{multline*}
    with the convention $\hmap[0]{\thp, \parvar}(h_0, a_{0:- 1}, y_{1:0}) \eqdef h_0$.  
\end{lemma}

\begin{proof}[Proof of \Cref{lem:ump:explicit:form}]
We proceed by induction and assume that the claim holds true for $t \in \nsetpos$. By definition \eqref{eq:def:umap:comp}, 
\begin{align*}
    \lefteqn{\umap[t + 1]{\thp,\parvar}(h_0, u_0, a_{0:t}, y_{1:t + 1})(x_{t + 1})} \\
    &= \umap{\thp,\parvar}(\hmap[t]{\thp,\parvar}(h_0, a_{0:t - 1}, y_{1:t}), \umap[t]{\thp,\parvar}(u_0, a_{0:t - 1}, y_{1:t}), a_t, y_{t + 1}) \\
    &= \int \umap[t]{\thp,\parvar}(u_0, a_{0:t - 1}, y_{1:t})(x_t) + \grad_\parvar \log \vd{t|t + 1}(x_{t + 1}, x_t) \{ \hmap[t]{\thp, \parvar}(h_0, a_{0:t - 1}, y_{1:t})(x_t) + \addfelbo{t}(x_t, x_{t + 1}) \} \, \vk{t|t + 1}(x_{t + 1}, \rmd x_t).  
\end{align*}
Now, inserting the induction hypothesis into the right-hand side of the previous expression yields 
\begin{multline*}
   \umap[t + 1]{\thp,\parvar}(h_0, u_0, a_{0:t}, y_{1:t + 1})(x_{t + 1}) \\ 
   = \idotsint \left( u_0(x_0) + \sum_{s = 0}^t \grad_\parvar \log \vd{s|s + 1}(x_{s + 1}, x_s) \{ \hmap[s]{\thp, \parvar}(h_0, a_{0:s - 1}, y_{1:s}) + \addfelbo{s}(x_s, x_{s + 1}) \} \right) \prod_{s = 1}^t \vk{s|s + 1}(x_{s + 1}, \rmd x_s),  
\end{multline*}
which establishes the induction step.  

Finally, we note that the base case $t = 1$ holds true, since by definition \eqref{eq:def:umap:comp}
\begin{align*}
    \umap[1]{\thp,\parvar}(h_0, u_0, a_0, y_1)(x_1) &= \umap{\thp,\parvar}(h_0, u_0, a_0, y_1)(x_1) \\
    &= \int \left( u_0(x_0) + \grad_\parvar \log \vd{0|1}(x_1, x_0) \{ h_0(x_0) + \addfelbo{0}(x_0, x_{0 + 1}) \} \right) \, \vk{0|1}(x_1, \rmd x_0).  
\end{align*}
This completes the proof. 
\end{proof}

In the following, let $\beta(M)$ denote the Dobrushin coefficient of a Markov kernel $M$. 

\begin{lemma} \label{lem:bounded:dc:Q}
    Assume \ref{ass:bounded-pot}. Then for every $t \in \nset$, 
    $\beta(\vk{t|t + 1}) \leq 1 - \varepsilon$. 
\end{lemma}

\begin{proof}[Proof of \Cref{lem:bounded:dc:Q}]
    Pick arbitrarily $(x_{t + 1}, f) \in \Xsp \times \bmf(\Xfd)$ and write, using definition \eqref{eq:def:varbackwd} and \Cref{ass:bounded-pot}, 
    $$
    \vk{t|t + 1} f(x_{t + 1}) = \frac{\int f(x_t) \fwdpot{t}(x_t, x_{t + 1}) \, \vd{t}(\rmd x_t)}{\int \fwdpot{t}(x'_t, x_{t + 1}) \, \vd{t}(\rmd x'_t)} \geq \frac{\varepsilon^-}{\varepsilon^+} \vd{t} f = \varepsilon \vd{t} f, 
    $$
    which means that $\vk{t|t + 1}$ allows $\Xsp$ as a 1-small set with respect to $(\vd{t}, \varepsilon)$. From this it follows that $\beta(\vk{t|t + 1}) \leq 1 - \varepsilon$. 
\end{proof}

We are now ready to establish \Cref{hmap:vmap:contraction:lemma}.  

\begin{proof}[Proof of \Cref{hmap:vmap:contraction:lemma}]
    To establish (i),  let $(Q^\parvar_{s - 1|s})_{s = 1}^t$ denote the backward transition kernels associated with $(q^\parvar_{s - 1|s})_{s = 1}^t$. We may then write, for every $x_t \in \Xsp$,  
    $$
    \hmap[t]{\thp, \parvar}(h_0, y_{1:t})(x_t) - \hmap[t]{\thp, \parvar}(h_0', y_{1:t})(x_t) = Q^\parvar_{t - 1|t} \cdots Q^\parvar_{0|1} (h_0 - h'_0)(x_t).  
    $$
    Now, recall that for every $s$ and $h \in \bmf(\Xfd)$, 
    \begin{equation} \label{eq:osc:norm:bound}
    \oscnorm(Q^\parvar_{s - 1|s} h) \leq \beta(Q^\parvar_{s - 1|s}) \oscnorm(h), 
    \end{equation}
    where $\beta(Q^\parvar_{s - 1|s})$ is the Dobrushin  coefficient of $Q^\parvar_{s - 1|s}$, and iterating the bound \eqref{eq:osc:norm:bound} yields 
    \begin{equation} \label{eq:osc:norm:bound:iterated}    \oscnorm(Q^\parvar_{t - 1|t} \cdots Q^\parvar_{0|1} (h_0 - h'_0)) \leq \left( \prod_{s = 1}^t \beta(Q^\parvar_{s - 1|s}) \right) \oscnorm(h_0 - h'_0).  
    \end{equation}
    From this the claim (i) follows by \Cref{lem:bounded:dc:Q}. 
    
    To prove (ii), note that  
    \begin{align*}
    \vmap[t]{\thp, \parvar}(v_0, a_{0:t - 1}, y_{1:t})(x_t) - \vmap[t]{\thp, \parvar}(v_0', a_{0:t - 1}, y_{1:t})(x_t) &= Q^\parvar_{t - 1|t} \cdots Q^\parvar_{0|1} (v_0 - v'_0)(x_t) \\ 
    &= \hmap[t]{\thp, \parvar}(v_0, a_{0:t - 1}, y_{1:t})(x_t) - \hmap[t]{\thp, \parvar}(v_0', a_{0:t - 1}, y_{1:t})(x_t). 
    \end{align*}
    Thus, (ii) follows immediately from (i).

    Finally, to establish (iii), write, using \Cref{lem:ump:explicit:form}, 
    \begin{align}
        \lefteqn{\umap[t]{\thp,\parvar}(h_0, u_0, a_{0:t - 1}, y_{1:t})(x_t) - \umap[t]{\thp,\parvar}(h'_0, u'_0, a_{0:t - 1}, y_{1:t})(x_t)} \nonumber \\ 
        &= \idotsint \left( u_0(x_0) - u'_0(x_0) + \sum_{s = 0}^{t - 1} \grad_\parvar \log \vd{s|s + 1}(x_{s + 1}, x_s) \{ \hmap[s]{\thp, \parvar}(h_0, a_{0:s - 1}, y_{1:s})(x_s) - \hmap[s]{\thp, \parvar}(h'_0, a_{0:s - 1}, y_{1:s})(x_s) \} \right) \nonumber \\
        & \hspace{130mm} \times \prod_{s = 0}^{t - 1} \vk{s|s + 1}(x_{s + 1}, \rmd x_s) 
        \nonumber \\
        &= \vk{t - 1|t} \cdots \vk{0|1}(u_0 - u'_0)(x_t) + \sum_{s = 0}^{t - 1} \vk{t - 1|t} \cdots \vk{s + 1|s + 2} \varphi_s(x_t), \label{eq:umap:diff:alt:form}
    \end{align}
    where we have set 
    $$
    \varphi_s(x_{s + 1}) \eqdef \int \grad_\parvar \log \vd{s|s + 1}(x_{s + 1}, x_s) \{ \hmap[s]{\thp, \parvar}(h_0, a_{0:s - 1}, y_{1:s})(x_s) - \hmap[s]{\thp, \parvar}(h'_0, a_{0:s - 1}, y_{1:s})(x_s) \} \, \vk{s|s + 1}(x_{s + 1}, \rmd x_s). 
    $$
    Now, note that since 
    $$
    \int \grad_\parvar \log \vd{s|s + 1}(x_{s + 1}, x_s) \, \vk{s|s + 1}(x_{s + 1}, \rmd x_s) = 0, 
    $$
    it holds, for every $c \in \rset$, 
    \begin{equation*} \label{eq:varphi:bd:1}
    \| \varphi_s \|_\infty  
    \leq \| \hmap[s]{\thp, \parvar}(h_0, a_{0:s - 1}, y_{1:s}) - \hmap[s]{\thp, \parvar}(h'_0, a_{0:s - 1}, y_{1:s}) - c\|_\infty \int | \grad_\parvar \log \vd{s|s + 1}(x_{s + 1}, x_s) | \vk{s|s + 1}(x_{s + 1}, \rmd x_s).
    \end{equation*}
    Thus, using \Cref{ass:grad:unif:bound}(i) and the fact that for all $f \in \mf(\Xfd)$, $\oscnorm(f) = 2 \inf_{c \in \rset}\| f - c\|_\infty$, it holds, by (i), that 
    $$
    \| \varphi_s \|_\infty \leq \frac{1}{2} c \oscnorm \left( \hmap[s]{\thp, \parvar}(h_0, a_{0:s - 1}, y_{1:s}) - \hmap[s]{\thp, \parvar}(h'_0, a_{0:s - 1}, y_{1:s}) \right) \leq \frac{1}{2} c (1 - \varepsilon)^s \oscnorm(h_0 - h_0'). 
    $$
    As a consequence, by \Cref{lem:bounded:dc:Q}, 
    \begin{align} 
        \oscnorm(\vk{t - 1|t} \cdots \vk{s + 1|s + 2} \varphi_s) &\leq \left( \prod_{\ell = s + 1}^{t - 1} \beta(\vk{\ell|\ell + 1}) \right) \oscnorm(\varphi_s) \nonumber \\
        &\leq c (1 - \varepsilon)^{t - s - 1} (1 - \varepsilon)^s \oscnorm(h_0 - h_0') \nonumber \\
        &= c (1 - \varepsilon)^{t - 1} \oscnorm(h_0 - h_0'). \label{eq:umap:diff:second:term:bd} 
    \end{align}
    Moreover, since 
    $$
    \vk{t - 1|t} \cdots \vk{0|1}(u_0 - u'_0)(x_t) = \hmap[t]{\thp, \parvar}(u_0, a_{0:t - 1}, y_{1:t})(x_t) - \hmap[t]{\thp, \parvar}(u_0', a_{0:t - 1}, y_{1:t})(x_t), 
    $$
    (i) implies that 
    \begin{equation} \label{eq:umap:diff:first:term:bd}
        \oscnorm(\vk{t - 1|t} \cdots \vk{0|1}(u_0 - u'_0)) \leq (1 - \varepsilon)^t \oscnorm(u_0 - u'_0). 
    \end{equation}
    Combining \eqref{eq:umap:diff:alt:form}, \eqref{eq:umap:diff:second:term:bd}, and \eqref{eq:umap:diff:first:term:bd} yields 
    $$
    \oscnorm \left( \umap[t]{\thp,\parvar}(h_0, u_0, a_{0:t - 1}, y_{1:t}) - \umap[t]{\thp,\parvar}(h'_0, u'_0, a_{0:t - 1}, y_{1:t}) \right) \leq (1 - \varepsilon)^t \oscnorm(u_0 - u'_0) + c t  (1 - \varepsilon)^{t - 1} \oscnorm(h_0 - h_0'). 
    $$
    Finally, the claim (iii) follows by picking $\varrho \in (1 - \varepsilon, 1)$ and letting $d \eqdef ((1 - \varepsilon) \operatorname{e} \log\{\varrho / (1 - \varepsilon)\})^{-1}$. 
\end{proof}

We are now ready to establish \Cref{prop:K_2:ergodicity}, following the same lines as \Cref{prop:K:ergodicity}. 

\begin{proof}[Proof of \Cref{prop:K_2:ergodicity}]
First, denote
\begin{multline*}
    f_{s,t}(x_t, y_{s:t}, a_{s - 1:t}, h, u, v) \\ 
    \eqdef f(x_t, y_t, a_t, \hmap[t - s + 1]{\thp,\parvar}(h, a_{s - 1:t - 1}, y_{s:t}), \umap[t - s + 1]{\thp,\parvar}(h, u, a_{s - 1:t - 1}, y_{s:t}), \vmap[t - s + 1]{\thp,\parvar}(v, a_{s - 1:t - 1}, y_{s:t})). 
\end{multline*}
for $s \in \intset{1}{t}$, where we have omitted the dependence on $\thp$ and $\parvar$ for brevity. Note that with this notation, for $z = (x_0, y_0, a_0, h_0, u_0, v_0)$, 
\begin{equation*}
    (\K)^t f(z) \\
    = \idotsint  f_{1,t}(x_t, y_{1:t}, a_{0:t}, h_0, u_0, v_0)  \prod_{s = 0}^{t - 1} \Km((x_s, y_s, a_s), \rmd(x_{s + 1}, y_{s + 1}, a_{s + 1})). 
\end{equation*}
Now, picking $(\tilde{h}_0, \tilde{u}_0, \tilde{v}_0) \in \bmf(\Xfd)^3$ arbitrarily and using the decomposition 
\begin{multline*}
f_{1,t}(x_t, y_{1:t}, a_{0:t}, h_0, u_0, v_0) = f_{1,t}(x_t, y_{1:t}, a_{0:t}, h_0, u_0, v_0) - f_{1,t}(x_t, y_{1:t}, a_{0:t}, \tilde{h}_0, \tilde{u}_0, \tilde{v}_0) \\ 
+ \sum_{s = 1}^{t - 1} \left( f_{s,t}(x_t, y_{s:t}, a_{s - 1:t}, \tilde{h}_0, \tilde{u}_0, \tilde{v}_0) - f_{s + 1,t}(x_t, y_{s + 1:t}, a_{s:t}, \tilde{h}_0, \tilde{u}_0, \tilde{v}_0) \right) + f_{t, t}(x_t, y_t, a_{t - 1:t}, \tilde{h}_0, \tilde{u}_0, \tilde{v}_0), 
\end{multline*}
which is adopted from \cite{tadic-doucet-2005}, we may write, 
\begin{align}
    \lefteqn{(\K)^t f(z) - (\K)^t f(z')} \nonumber \\ 
    &= \idotsint \left( f_{1,t}(x_t, y_{1:t}, a_{0:t}, h_0, u_0, v_0) - f_{1,t}(x_t, y_{1:t}, a_{0:t}, \tilde{h}_0, \tilde{u}_0, \tilde{v}_0) \right) \prod_{s = 0}^{t - 1} \S((x_s, y_s, a_s), \rmd(x_{s + 1}, y_{s + 1}, a_{s + 1})) \nonumber \\
    & - \idotsint \left( f_{1,t}(x'_t, y'_{1:t}, a'_{0:t}, h'_0, u'_0, v'_0) - f_{1,t}(x'_t, y'_{1:t}, a'_{0:t}, \tilde{h}_0, \tilde{u}_0, \tilde{v}_0)  \right) \prod_{s = 0}^{t - 1} \S((x'_s, y'_s, a'_s), \rmd(x'_{s + 1}, y'_{s + 1}, a'_{s + 1})) \nonumber \\
    &+ \sum_{s = 1}^t \idotsint \left( f_{s,t}(x_t, y_{s:t}, a_{s - 1:t}, \tilde{h}_0, \tilde{u}_0, \tilde{v}_0) - f_{s + 1,t}(x_t, y_{s + 1:t}, a_{s:t}, \tilde{h}_0, \tilde{u}_0, \tilde{v}_0) \right) \nonumber \\
    &\hspace{60mm} \times 
    \Delta \S_s((x_0, y_0, a_0), \rmd(x_s, y_s, a_s)) \, \prod_{\ell = s}^{t - 1} \S((x_\ell, y_\ell, a_\ell), \rmd(x_{\ell + 1}, y_{\ell + 1}, a_{\ell + 1})) \nonumber \\
    &- \sum_{s = 1}^t \idotsint \left( f_{s,t}(x_t, y_{s:t}, a_{s - 1:t}, \tilde{h}_0, \tilde{u}_0, \tilde{v}_0) - f_{s + 1,t}(x_t, y_{s + 1:t}, a_{s:t}, \tilde{h}_0, \tilde{u}_0, \tilde{v}_0) \right) \nonumber \\
    &\hspace{60mm} \times 
    \Delta \S_s((x'_0, y'_0, a'_0), \rmd(x_s, y_s, a_s)) \, \prod_{\ell = s}^{t - 1} \S((x_\ell, y_\ell, a_\ell), \rmd(x_{\ell + 1}, y_{\ell + 1}, a_{\ell + 1})) \nonumber \\
    &+ \int f_{t, t}(x_t, y_t, a_{t - 1:t}, \tilde{h}_0, \tilde{u}_0, \tilde{v}_0) \, \Delta \S_t((x_0, y_0, a_0), \rmd(x_t, y_t, a_t)) \nonumber \\
    &- \int f_{t, t}(x_t, y_t, a_{t - 1:t}, \tilde{h}_0, \tilde{u}_0, \tilde{v}_0)
    \, \Delta \S_t((x'_0, y'_0, a'_0), \rmd(x_t, y_t, a_t)), \label{eq:T:diff:decomp}
\end{align}
where we have defined, for $s \in \intset{1}{t}$, the signed kernels 
$$
\Delta \S_s f(x, y, a) \eqdef (\S)^s f(x, y, a) - \pi f, \quad (f, x, y, a) \in \bmf(\Xfd \tensprod \Yfd \tensprod \Afd) \times \Xsp \times \Ysp \times \Asp.   
$$
Note that by \Cref{def:lip:h}, for every $s \in \intset{1}{t}$, 
\begin{align}
    \lefteqn{|f_{s,t}(x_t, y_{s:t}, a_{s - 1:t}, \tilde{h}_0, \tilde{u}_0, \tilde{v}_0) - f_{s + 1,t}(x_t, y_{s + 1:t}, a_{s:t}, \tilde{h}_0, \tilde{u}_0, \tilde{v}_0)|} \nonumber \\ 
    &\leq \varphi(x_t, y_t, a_t) \left( \oscnorm(\hmap[t - s + 1]{\thp,\parvar}(\tilde{h}_0, a_{s - 1:t - 1}, y_{s:t}) - \hmap[t - s]{\thp,\parvar}(\tilde{h}_0, a_{s:t - 1}, y_{s + 1:t})) \right. \nonumber \\
    &\hspace{21mm}+ \oscnorm \left( \umap[t - s + 1]{\thp,\parvar}(\tilde{h}_0, \tilde{u}_0, a_{s - 1:t - 1}, y_{s:t}) - \umap[t - s]{\thp,\parvar}(\tilde{h}_0, \tilde{u}_0, a_{s:t - 1}, y_{s + 1:t}) \right) \nonumber \\
    &\hspace{21mm}+ \left. \oscnorm \left( \vmap[t - s + 1]{\thp,\parvar}(\tilde{v}_0, a_{s - 1:t - 1}, y_{s:t}) - \vmap[t - s]{\thp,\parvar}(\tilde{v}_0, a_{s:t - 1}, y_{s + 1:t}) \right) \right). \label{eq:f:diff:bd}
\end{align}
Using definition \eqref{eq:def:hmap:comp} and \Cref{hmap:vmap:contraction:lemma}(i), we conclude that 
\begin{align*}
    \lefteqn{\oscnorm\left(\hmap[t - s + 1]{\thp,\parvar}(\tilde{h}_0, a_{s - 1:t - 1}, y_{s:t}) - \hmap[t - s]{\thp,\parvar}(\tilde{h}_0, a_{s:t - 1}, y_{s + 1:t})\right)} \hspace{20mm} \\
    &= \oscnorm\left(\hmap[t - s]{\thp,\parvar}(\hmap{\thp,\parvar}(\tilde{h}_0, y_s), a_{s:t - 1}, y_{s + 1:t}) - \hmap[t - s]{\thp,\parvar}(\tilde{h}_0, a_{s:t - 1}, y_{s + 1:t})\right) \\
    &\leq (1 - \varepsilon)^{t - s} \oscnorm \left( \hmap{\thp,\parvar}(\tilde{h}_0, y_s) - \tilde{h}_0 \right).  
\end{align*}
Since 
\begin{align*}
    \| \hmap{\thp,\parvar}(\tilde{h}_0, y_s) - \tilde{h}_0 \|_\infty &= \left\| \int (\tilde{h}_0(x) + \afelbo{s - 1}(x, \cdot)) \, \vk{s - 1|s}(\cdot, \rmd x) - \tilde{h}_0(\cdot) \right\|_\infty \\
    &\leq 2 \| \tilde{h}_0 \|_\infty + c, 
\end{align*}
where the constant $c \in \rsetpos$ is provided by \Cref{ass:grad:unif:bound}(ii), it holds that  
\begin{equation} \label{eq:bound:h:diff}
    \oscnorm\left(\hmap[t - s + 1]{\thp,\parvar}(\tilde{h}_0, a_{s - 1:t - 1}, y_{s:t}) - \hmap[t - s]{\thp,\parvar}(\tilde{h}_0, a_{s:t - 1}, y_{s + 1:t})\right) \leq 2(2 \| \tilde{h}_0 \|_\infty + c) (1 - \varepsilon)^{t - s} .  
\end{equation}
By the same arguments it is shown that  
\begin{equation}\label{eq:bound:v:diff}   \oscnorm\left(\vmap[t - s + 1]{\thp,\parvar}(\tilde{v}_0, a_{s - 1:t - 1}, y_{s:t}) - \vmap[t - s]{\thp,\parvar}(\tilde{v}_0, a_{s:t - 1}, y_{s + 1:t})\right) \leq 2(2 \| \tilde{v}_0 \|_\infty + c) (1 - \varepsilon)^{t - s}.
\end{equation}
In addition, similarly, using \Cref{hmap:vmap:contraction:lemma}(iii), 
\begin{align*} \lefteqn{\oscnorm\left(\umap[t - s + 1]{\thp,\parvar}(\tilde{h}_0, \tilde{u}_0, a_{s - 1:t - 1}, y_{s:t}) - \umap[t - s]{\thp,\parvar}(\tilde{h}_0, \tilde{u}_0, a_{s:t - 1}, y_{s + 1:t})\right)} \hspace{20mm} \\
    &= \oscnorm\left(\umap[t - s]{\thp,\parvar}(\hmap{\thp,\parvar}(\tilde{h}_0, y_s), \umap{\thp,\parvar}(\tilde{h}_0, \tilde{u}_0, y_s), a_{s:t - 1}, y_{s + 1:t}) - \umap[t - s]{\thp,\parvar}(\tilde{h}_0, \tilde{u}_0, a_{s:t - 1}, y_{s + 1:t})\right) \\
    &\leq (1 - \varepsilon)^{t - s} \oscnorm \left( \hmap{\thp,\parvar}(\tilde{h}_0, y_s) - \tilde{h}_0 \right) + d \varrho^{t - s} \oscnorm \left( \umap{\thp,\parvar}(\tilde{h}_0, \tilde{u}_0,  y_s) - \tilde{u}_0 \right), 
\end{align*}
and since by \Cref{ass:grad:unif:bound}(i--ii) and the Cauchy--Schwarz inequality, 
\begin{align*}
    \left\| \umap{\thp,\parvar}(\tilde{h}_0, \tilde{u}_0, y_s) - \tilde{u}_0 \right\|_\infty &= \left\| \int (\tilde{u}_0(x) + \grad_\parvar \log \vd{s - 1|s}(\cdot, x) \{ \tilde{h}_0(x) + \afelbo{s - 1}(x, \cdot)\}) \, \vk{s - 1|s}(\cdot, \rmd x) - \tilde{h}_0(\cdot) \right\|_\infty \\
    &\leq 2 \| \tilde{u}_0 \|_\infty + c \| \tilde{h}_0 \|_\infty + c^2, 
\end{align*}
we may conclude that 
\begin{multline} \label{eq:bound:u:diff}
\oscnorm\left(\umap[t - s + 1]{\thp,\parvar}(\tilde{h}_0, \tilde{u}_0, a_{s - 1:t - 1}, y_{s:t}) - \umap[t - s]{\thp,\parvar}(\tilde{h}_0, \tilde{u}_0, a_{s:t - 1}, y_{s + 1:t})\right) \\ 
\leq 2(2 \| \tilde{h}_0 \|_\infty + c) (1 - \varepsilon)^{t - s} + 2 d (2 \| \tilde{u}_0 \|_\infty + c \| \tilde{h}_0 \|_\infty + c^2) \varrho^{t - s}. 
\end{multline}
Combining \eqref{eq:f:diff:bd}, \eqref{eq:bound:h:diff}, \eqref{eq:bound:v:diff}, and  \eqref{eq:bound:u:diff} yields the bound  
\begin{align}
    \lefteqn{|f_{s,t}(x_t, y_{s:t}, a_{s - 1:t}, \tilde{h}_0, \tilde{u}_0, \tilde{v}_0) - f_{s + 1,t}(x_t, y_{s + 1:t}, a_{s:t}, \tilde{h}_0, \tilde{u}_0, \tilde{v}_0)|} \hspace{30mm}
    \nonumber \\ 
    &\leq 2\varphi(x_t, y_t, a_t) \left( 3(2 \| \tilde{h}_0 \|_\infty + c) (1 - \varepsilon)^{t - s} +  d (2 \| \tilde{u}_0 \|_\infty + c \| \tilde{h}_0 \|_\infty + c^2) \varrho^{t - s} \right) \nonumber \\
    & \leq  2 \varphi(x_t, y_t, a_t)  \left( (6 + cd ) \| \tilde{h}_0 \|_\infty + 2d\| \tilde{u}_0 \|_\infty + 3 c  +  c^2 d \right) \varrho^{t - 2}. \label{eq:f:different:s:bound}
\end{align}
Similarly, 
\begin{equation} \label{eq:f:same:t:bound}
    |f_{1,t}(x_t, y_{1:t}, a_{0:t}, h_0, u_0, v_0) - f_{1,t}(x_t, y_{1:t}, a_{0:t}, \tilde{h}_0, \tilde{u}_0, \tilde{v}_0)| 
    \leq  \varphi(x_t, y_t, a_t) \statemap''(h_0, u_0, v_0) \varrho^t,   
\end{equation}
where we have defined the mapping  
$$
\statemap''(h_0, u_0, v_0) \eqdef 2 \left( (d + 1) (\| h_0 \| + \| \tilde{h}_0 \|_\infty ) + \| u_0 \| + \| \tilde{u}_0 \|_\infty + \| v_0 \| + \| \tilde{v}_0 \|_\infty \right).  
$$
Now, letting 
\begin{equation*}
    \statemap'(h_0, h'_0, u_0, u'_0, v_0, v'_0) \eqdef 2 \left( (6 + cd ) \| \tilde{h}_0 \|_\infty + 2d\| \tilde{u}_0 \|_\infty + 3 c  +  c^2 d \right) \vee \statemap''(h_0, u_0, v_0) \vee \statemap''(h'_0, u'_0, v'_0).  
\end{equation*}
Applying the bounds  \eqref{eq:f:different:s:bound} and \eqref{eq:f:same:t:bound} to the decomposition \eqref{eq:T:diff:decomp} yields  
\begin{multline*} \label{eq:K:diff:bound:1}
    |(\K)^t f(z) - (\K)^t f(z')|  
    \leq \statemap'(h_0, h'_0, u_0, u'_0, v_0, v'_0) \varrho^t \left( (\S)^t \varphi (x_0, y_0, a_0) + (\S)^t \varphi (x'_0, y'_0, a'_0) \right) \\ 
    + \statemap'(h_0, h'_0, u_0, u'_0, v_0, v'_0) \sum_{s = 1}^t \varrho^{t - s} \left( \Delta \S_s (\S)^{t - s} \varphi(x_0, y_0, a_0) + \Delta \S_s (\S)^{t - s} \varphi(x'_0, y'_0, a'_0) \right) \\
    + \Delta \S_t \varphi (x_0, y_0, a_0) + \Delta \S_t \varphi (x'_0,  y'_0, a'_0). 
\end{multline*}
Now, by applying \Cref{ass:unif:erg:xya} to the right-hand side of \eqref{eq:K:diff:bound:1} we obtain 
\begin{equation*}
  |(\K)^t f(z) - (\K)^t f(z')| \leq \| \varphi \|_\infty \left(\statemap'(h_0, h'_0, u_0, u'_0, v_0, v'_0)\alpha^t + \statemap'(h_0, h'_0, u_0, u'_0, v_0, v'_0) t (\alpha \vee \varrho)^t + \varrho^t \right),    
\end{equation*}
from which \eqref{eq:K:diff:bound:final:th} follows by picking $\rho \in (\alpha \vee \varrho, 1)$ and 
$$
\statemap(h_0, h'_0, u_0, u'_0, v_0, v'_0) \eqdef 2(\statemap'(h_0, h'_0, u_0, u'_0, v_0, v'_0) + 1) (\alpha \vee \varrho)/\rho + 2 \statemap'(h_0, h'_0, u_0, u'_0, v_0, v'_0) \left( \operatorname{e} |\log\{(\alpha \vee \varrho)/\rho\} |\right)^{-1}.    
$$
To prove the second claim, first note that by \eqref{eq:K:diff:bound:final:th}, 
\begin{align}
|(\K)^{t + 1} f(z) - (\K)^t f(z)| 
&\leq \int |(\K)^t f(z') - (\K)^t f(z)| \, \K(z, \rmd z') \nonumber \\
&\leq \| \varphi \|_\infty \rho^t \int \statemap(h_0, h'_0, u_0, u'_0, v_0, v'_0) \, \K(z, \rmd z').  
\label{eq:K:one:step:diff}
\end{align}
Now, define the kernel 
$$
\Pi^{\thp, \parvar} f(z) \eqdef f(z) + \sum_{t = 0}^\infty \left( (\K)^{t + 1} f(z) - (\K)^t f(z) \right)  
$$
on $\Zsp \times \lip(\Zfd)$. 
With this definition, note that by \eqref{eq:K:one:step:diff}, 
\begin{align*}
    |(\K)^t f(z) - \Pi^{\thp, \parvar} f(z)| &\leq \sum_{s = t}^\infty |(\K)^{s + 1} f(z) - (\K)^s f(z)| \nonumber \\
    &\leq \bar{\statemap}(h_0, u_0, v_0) \| \varphi \|_\infty \rho^t,  
\end{align*}
where $\bar{\statemap}(h_0, u_0, v_0)$ is defined in \eqref{eq:def:statemap:bar}, 
which establishes \eqref{eq:K:Sigma:diff}. 
Finally, it remains to prove that the function $\Pi^{\thp, \parvar} f$ is constant. For this purpose, pick arbitrarily $(z, z') \in \Zsp^2$; then, however, by \eqref{eq:K:diff:bound:final:th} and \eqref{eq:K:Sigma:diff},  
\begin{multline*}
|\Pi^{\thp, \parvar} f(z) - \Pi^{\thp, \parvar} f(z')| \leq \inf_{t \in \nset} \left( |(\K)^{t + 1} f(z) - (\K)^t f(z)| + |(\K)^t f(z) - \Pi^{\thp, \parvar} f(z) | \right. \\ + \left. |(\K)^t f(z') - \Pi^{\thp, \parvar} f(z')| \right) = 0, 
\end{multline*}
from which the claim follows.
\end{proof}

\subsection{Discussion on \Cref{ass:unif:erg:xya}}
\label{sec:discussion:unif:erg:xya}

Recall \Cref{ass:unif:erg:xya}, under which it is supposed that for every $\parvar \in \Phi$ there exist $\alpha \in (0, 1)$ and $\pi^{\parvar} \in \probmeas(\Xfd \tensprod \Yfd \tensprod \Afd)$ ($\alpha$ being independent of $\parvar$) such that for every $t \in \nset$, $(x, y, a) \in \Xsp \times \Ysp \times \Asp$, and $f \in \bmf(\Xfd \tensprod \Yfd \tensprod \Afd)$, 
\begin{equation} \label{eq:TV:ass:alt}
|(\S)^t f(x, y, a) - \pi^\parvar f| \leq \| f \|_\infty \alpha^t, 
\end{equation}
where 
\begin{multline*}
\S f(x_t, y_t, a_t) \eqdef \int f(x_{t + 1}, y_{t + 1}, \amap{\parvar}(a_t, y_{t + 1})) \, S((x_t, y_t), \rmd (x_{t + 1}, y_{t + 1})), \\ (x_t, y_t, a_t, f) \in \Xsp \times \Ysp \times \Asp \times \bmf(\Xfd \tensprod \Yfd \tensprod \Afd),   
\end{multline*}
is the transition kernel of the Markov chain $(X_t, Y_t, a_t)_{t \in \nset}$. 

In this section, we will present conditions under which \eqref{eq:TV:ass:alt} applies, if not for all $f \in \bmf(\Xfd \tensprod \Yfd \tensprod \Afd)$, then at least 
for all $f$ in a  certain Lipschitz subclass $\liptd(\Xfd \tensprod \Yfd \tensprod \Afd)$ of $\mf(\Xfd \tensprod \Yfd \tensprod \Afd)$ to be specified. 
Recall the mapping $\amap{\parvar}$ introduced in \Cref{sec:model:background} and define, for every $t \in \nsetpos$ and vectors $y_{1:t} \in \Ysp^t$ and $a_{0:t - 1} \in \Asp^t$, the composite versions 
\begin{equation} \label{eq:def:amap:comp}
\amap[t]{\parvar}(\cdot, y_{1:t}) \eqdef 
\begin{cases}
     \amap{\parvar}(\cdot, y_1), & \mbox{for $t = 1$,} \\
     \amap{\parvar}(\amap[t - 1]{\parvar}(\cdot, y_{1:t - 1}), y_t), & \mbox{for $t \geq 2$.}
\end{cases}
\end{equation}
We will assume that these satisfy the following assumption. 

\begin{assumption} \label{ass:amap:contraction}
    There exist $c > 0$ and $\kappa \in (0, 1)$ such that for every $\parvar \in \Phi$, $t \in \nset$, $y_{1:t} \in \Ysp^t$, 
    and $(a, \tilde{a}) \in \Asp^2$, 
    \begin{itemize}
    \item[(i)]
    $\displaystyle
    \| \amap[t]{\parvar}(a, y_{1:t}) - \amap[t]{\parvar}(\tilde{a}, y_{1:t}) \|_2 \leq c \|a - \tilde{a}\|_2\kappa^t$, 
    \item[(ii)] $\displaystyle \| \amap{\parvar}(a, y_t) - a \|_2 \leq c$.
    \end{itemize}
\end{assumption}

\begin{remark}
    To illustrate \Cref{ass:amap:contraction}(i), let $\amap{\parvar}$ correspond to the following vanilla recurrent neural network architecture, where $a_0 \in \rset^d$ is the initial hidden state. For $t \in \nsetpos$, given a sequence $y_{1:t}$ of observations, the hidden state $a_t$ is given by
    $$
        a_t = \amap{\parvar}(a_{t - 1}, y_t) = \sigma(W_\parvar a_{t-1} + U_\parvar y_t + b_\parvar),
    $$
where
\begin{itemize}
    \item $W_\parvar \in \rset^{d\times d}$ and $U_\parvar \in \rset^{d\times d_y}$ are  weight matrices and $b_\parvar \in \rset^d$ is a bias vector.
    \item $\sigma : \rset^d \rightarrow \rset^d$ is a Lipschitz mapping with Lipschitz constant $L_\sigma$. Typical examples include the ReLU and hyperbolic tangent  (tanh) functions, both of which satisfy $L_\sigma = 1$.
\end{itemize}
Consider the spectral norm of $W_\parvar$ given by
$$
L_{W_\parvar} = \|W_\parvar\|_2 = \sup_{\|x\|_2=  1} \|W_\parvar x\|_2, 
$$
\emph{i.e.}, $L_{W_\parvar}$ is the largest singular value of $W_\parvar$.
Now, let us consider  two sequences $(a_t)_{t \in \nset}$ and $(\tilde{a}_t)_{t \in \nset}$ of hidden states starting with $a_0$ and $\tilde{a}_0$, respectively.
For every fixed observation sequence $y_{1:t}$, it holds, by definition  \eqref{eq:def:amap:comp}, that
\begin{align*}
     \|\amap[t]{\parvar}(a_0, y_{1:t}) - \amap[t]{ \parvar}(\tilde{a}_0, y_{1:t}) \|_2 &\leq L_\sigma\left\|W\left(\amap[t-1]{\parvar}(a_0, y_{1:t-1})-\amap[t-1]{\parvar}(\tilde{a}_0, y_{1:t-1})\right)\right\|_2\\
     &\leq L_\sigma^t \|W\|_2^t \|a_0 - \tilde{a}_0\|_2.
\end{align*}
Typically, we may assume here that $L_\sigma = 1$ (again, this is the case for the ReLU and tanh functions). \Cref{ass:amap:contraction}(i) then holds if $\|W_\parvar\|_2< 1$. 
    This contraint can be enforced during training via spectral norm regularization \cite{yoshida2017spectral}. Alternatively, one may enforce it \emph{a priori} by fixing the spectral norm to be at most $\rho_{\text{max}} < 1$ and, at each gradient step $k$,  projecting the unconstrained weight matrix  $\tilde{W}_\parvar^{(k)}$ onto the corresponding spectral-norm ball according to  
    $$
    W_\parvar^{(k)} = \rho_{\text{max}}  \frac{\tilde{W}_\parvar^{(k)}}{
    \| \tilde{W}_\parvar^{(k)}\|_2}.
    $$
    Such projections are already standard in echo state networks \cite{sun2020review}.
\end{remark}

\begin{assumption}
[Uniform ergodicity of $(X_t, Y_t)_{t \in \nset}$]\label{ass:ergodicity:SSM} There exist $\beta \in (0, 1)$ and  
$\sigma \in \probmeas(\Xfd \tensprod \Yfd)$ such that for every $(x,y) \in \Xsp \times \Ysp$ and $t\in \nsetpos$,
\[
  \| \R^t((x,y),\cdot) - \sigma \|_{\mathsf{TV}}
  \le \beta^t.
\]
\end{assumption}

In addition, consider the following Lipschitz class of functions. 
\begin{definition}
\label{def:lip:h:tilde}
    Let $\tilde{\lip}(\Xfd \tensprod \Yfd \tensprod \Afd)$ be the set of $f \in \mf(\Xfd \tensprod \Yfd \tensprod \Afd)$ for which there exists $\tilde{\varphi} \in \bmf(\Xfd \tensprod \Yfd)$ such that for every $(x, y, a, a') \in \Xsp \times \Ysp \times \Asp^2$, 
    \begin{itemize}
        \item[(i)] $|f(x, y, a)| \leq \tilde{\varphi}(x, y)$, 
        \item[(ii)] $|f(x, y, a) - f(x, y, a')| \leq \tilde{\varphi}(x, y) \|a - a'\|_2$.  
    \end{itemize}
\end{definition}

\begin{proposition} \label{prop:K:ergodicity}
    Assume \ref{ass:amap:contraction} and \ref{ass:ergodicity:SSM}. Then there exist a function $\delta : \Asp^2 \to \rsetpos$ and $\alpha \in (0, 1)$ such that for every $\parvar \in \vparsp$, $t \in \nset$, $f \in \tilde{\lip}(\Xfd \tensprod \Yfd \tensprod \Afd)$, and $((x, y, a), (x', y', a')) \in (\Xsp \times \Ysp \times \Asp)^2$,  
    \begin{equation}    \label{eq:K:diff:bound:final:s}
    |(\S)^t f(x, y, a) - (\S)^t f(x', y', a')| \leq \delta(a, a') \| \tilde{\varphi} \|_\infty \alpha^t.     
    \end{equation}
    Moreover, there exists a kernel $\pi^{\parvar}$ on $(\Xsp \times \Ysp \times \Asp) \times \tilde{\lip}(\Xfd \tensprod \Yfd \tensprod \Afd)$ such that for every $f \in \lip(\Xfd \tensprod \Yfd \tensprod \Afd)$, $\pi^{ \parvar} f$ is constant and for every $t \in \nset$ and $(x, y, a) \in \Xsp \times \Ysp \times \Asp$, 
    \begin{equation} \label{eq:K:Sigma:diff:s}
    |(\S)^t f(x, y, a) - \pi^{\parvar} f| \leq \bar{\delta}(a) \| \tilde{\varphi} \| \alpha^t,  
    \end{equation}
    where 
    \begin{equation} \label{eq:def:delta:bar}
    \bar{\delta}(a) \eqdef \frac{1}{1 - \alpha} \int \delta(a, a') \, \S((x, y, a), \rmd(x',y',a')).  
    \end{equation}
\end{proposition}

\begin{proof}[Proof of \Cref{prop:K:ergodicity}]
We proceed as in the proof of \Cref{prop:K_2:ergodicity}. 
To prove the first claim, we introduce the short-hand notation 
\begin{equation} \label{eq:f:amap:comp:def}
f_{s,t}(x_t, y_{s:t}, a) \eqdef f(x_t, y_t, \amap[t - s + 1]{\parvar}(a, y_{s:t}))
\end{equation}
for $s \in \intset{1}{t}$, where we have omitted the dependence on $\parvar$ for brevity. Note that with this notation, 
\begin{equation*}
    (\S)^t f(z) \\
    = \idotsint  f_{1,t}(x_t, y_{1:t}, a) \, \R((x, y), \rmd(x_1, y_1)) \prod_{s = 1}^{t - 1} \R((x_s, y_s), \rmd(x_{s + 1}, y_{s + 1})). 
    \end{equation*}
Now, picking $\tilde{a} \in \Asp$ arbitrarily and using the decomposition 
\begin{equation*}
f_{1,t}(x_t, y_{1:t}, a) = f_{1,t}(x_t, y_{1:t}, a) - f_{1,t}(x_t, y_{1:t}, \tilde{a}) 
+ \sum_{s = 1}^t \left( f_{s, t}(x_t, y_{s:t}, \tilde{a}) - f_{s + 1, t}(x_t, y_{s + 1:t}, \tilde{a})\right) + f_{t, t}(x_t, y_t, \tilde{a}), 
\end{equation*}
we may write, 
\begin{align}
    \lefteqn{(\S)^t f(x, y, a) - (\S)^t f(x', y', a')} \nonumber \\ 
    &= \idotsint \left( f_{1,t}(x_t, y_{1:t}, a) - f_{1,t}(x_t, y_{1:t}, \tilde{a})  \right) \R((x, y), \rmd(x_1, y_1)) \prod_{s = 1}^{t - 1} \R((x_s, y_s), \rmd(x_{s + 1}, y_{s + 1})) \nonumber \\
    & - \idotsint \left( f_{1,t}(x_t, y_{1:t}, a') - f_{1,t}(x_t, y_{1:t}, \tilde{a})  \right) \R((x', y'), \rmd(x_1, y_1)) \prod_{s = 1}^{t - 1} \R((x_s, y_s), \rmd(x_{s + 1}, y_{s + 1})) \nonumber \\
    &+ \sum_{s = 1}^t \idotsint \left( f_{s, t}(x_t, y_{s:t}, \tilde{a}) - f_{s + 1, t}(x_t, y_{s + 1:t}, \tilde{a})\right) 
    \Delta \R_s((x, y), \rmd(x_s, y_s)) \, \prod_{\ell = s}^{t - 1} \R((x_\ell, y_\ell), \rmd(x_{\ell + 1}, y_{\ell + 1})) \nonumber \\
    & - \sum_{s = 1}^t \idotsint \left( f_{s, t}(x_t, y_{s:t}, \tilde{a}) - f_{s + 1, t}(x_t, y_{s + 1:t}, \tilde{a})\right) 
    \Delta \R_s((x', y'), \rmd(x_s, y_s)) \, \prod_{\ell = s}^{t - 1} \R((x_\ell, y_\ell), \rmd(x_{\ell + 1}, y_{\ell + 1})) \nonumber \\
    &+ \int f_{t, t}(x_t, y_t, \tilde{a}) \, \Delta \R_t((x, y), \rmd(x_t, y_t)) - \int f_{t, t}(x_t, y_t, \tilde{a}) \, \Delta \R_t((x', y'), \rmd(x_t, y_t)), \label{eq:K:diff:decomp}
\end{align}
where we have defined, for $s \in \intset{1}{t}$, the signed kernels 
$$
\Delta \R_s f(x, y) \eqdef \R^s f(x, y) - \sigma f, \quad (f, x, y) \in \bmf(\Xfd \tensprod \Yfd) \times \Xsp \times \Ysp.   
$$
Note that by \Cref{def:lip:h:tilde} and \Cref{ass:amap:contraction}, for every $s \in \intset{1}{t}$,   
\begin{align}
    |f_{s, t}(x_t, y_{s:t}, \tilde{a}) - f_{s + 1, t}(x_t, y_{s + 1:t}, \tilde{a})| &\leq \tilde{\varphi}(x_t, y_t) \| \amap[t - s + 1]{\parvar}(\tilde{a}, y_{s:t}) - \amap[t - s]{\parvar}(\tilde{a}, y_{s + 1:t})\|_2 \nonumber \\
    &= \tilde{\varphi}(x_t, y_t) \| \amap[t - s]{\parvar}(\amap{\parvar}(\tilde{a}, y_s), y_{s + 1:t}) - \amap[t - s]{\parvar}(\tilde{a}, y_{s + 1:t})\|_2 \nonumber \\
    &\leq \tilde{\varphi}(x_t, y_t) c \|\amap{\parvar}(\tilde{a}, y_s) -  \tilde{a} \|_2 \kappa^{t - s} \nonumber \\
    &\leq \tilde{\varphi}(x_t, y_t) c^2 \kappa^{t - s}. \label{eq:h:bound:gen}
\end{align}
Similarly, 
\begin{align}
|f_{1,t}(x_t, y_{1:t}, a) - f_{1,t}(x_t, y_{1:t}, \tilde{a})| &\leq \tilde{\varphi}(x_t, y_t) c \| a - \tilde{a} \|_2 \nonumber \\
&\leq \tilde{\varphi}(x_t, y_t) c (\| a \|_2 + \| \tilde{a} \|) \kappa^t.  \label{eq:h:bound:spec}
\end{align}
Now, let 
$$
\delta'(a, a') \eqdef c^2 \vee c (\| a \|_2 + \| \tilde{a} \|) \vee c (\| a' \|_2 + \| \tilde{a} \|); 
$$ 
then, applying the bounds  \eqref{eq:h:bound:gen} and \eqref{eq:h:bound:spec} to the decomposition \eqref{eq:K:diff:decomp} yields 
\begin{multline} \label{eq:K:diff:bound:1:s}
    |(\S)^t f(x, y, a) - (\S)^t f(x', y', a')| \\ 
    \leq \delta'(a, a') \kappa^t \left( \R^t \tilde{\varphi} (x, y) + \R^t \tilde{\varphi} (x', y') \right) + \delta'(a, a') \sum_{s = 1}^t \kappa^{t - s} \left( \Delta \R_s \R^{t - s} \tilde{\varphi}(x, y) + \Delta \R_s \R^{t - s} \tilde{\varphi}(x', y') \right) \\
    + \Delta \R_t \tilde{\varphi} (x, y) + \Delta \R_t \tilde{\varphi} (x', y'). 
\end{multline}
Now, by applying \Cref{ass:ergodicity:SSM} to the right-hand side of \eqref{eq:K:diff:bound:1:s} we obtain 
\begin{equation*}
  |(\S)^t f(x, y, a) - (\S)^t f(x', y', a')| \leq 2 \| \tilde{\varphi} \|_\infty \left(\delta'(a, a') \kappa^t + \delta'(a, a') t (\kappa \vee \beta)^t + \beta^t \right),     
\end{equation*}
from which \eqref{eq:K:diff:bound:final:s} follows by picking $\alpha \in ((\kappa \vee \beta), 1)$ and letting 
$$
\delta(a, a') \eqdef 2(\delta'(a, a') + 1)(\kappa \vee \beta) / \alpha + 2\delta'(a, a')(\operatorname{e} |\log\{(\kappa \vee \beta) / \alpha\}|)^{-1}. 
$$ 
To prove the second claim, first note that by \eqref{eq:K:diff:bound:final:s}, \begin{align}
\lefteqn{|(\S)^{t + 1} f(x, y, a) - (\S)^t f(x, y, a)|} \nonumber \hspace{15mm}\\
&\leq \int |(\S)^t f(x', y', a') - (\S)^t f(x, y, a)| \, \S((x, y, a), \rmd(x', y', a')) \nonumber \\
&\leq \| \tilde{\varphi} \|_\infty \alpha^t \int \delta(a, a') \, \S((x, y, a), \rmd(x', y', a')). \label{eq:K:one:step:diff:s}
\end{align}
Now, define the kernel 
$$
\pi^{\parvar} f(x, y, a) \eqdef f(x, y, a) + \sum_{t = 0}^\infty \left( (\S)^{t + 1} f(x, y, a) - (\S)^t f(x, y, a) \right)
$$
on $\Xsp \times \Ysp \times \Asp \times \lip(\Xfd \tensprod \Yfd \tensprod \Afd)$. With this definition, note that by \eqref{eq:K:one:step:diff:s}, 
\begin{align*}
    |(\S)^t f(x, y, a) - \pi^{\parvar} f(x, y, a)| &\leq \sum_{s = t}^\infty |(\S)^{s + 1} f(x, y, a) - (\S)^s f(x, y, a)| \nonumber \\
    &\leq \delta(a) \| \tilde{\varphi} \|_\infty \alpha^t,  
\end{align*}
where $\delta(a)$ is provided by \eqref{eq:def:delta:bar}, 
which establishes \eqref{eq:K:Sigma:diff:s}. 

Finally, it remains to prove that the function $\pi^{\parvar} f$ is constant. For this purpose, pick arbitrarily $((x, y, a),(x'y',a')) \in (\Xsp \times \Ysp \times \Asp)^2$; then, however, by \eqref{eq:K:diff:bound:final:s} and \eqref{eq:K:Sigma:diff:s},  
\begin{multline*}
|\pi^{\parvar} f(x, y, a) - \pi^{\parvar} f(x', y', a')| \leq \inf_{t \in \nset} \left( |(\S)^{t + 1} f(x, y, a) - (\S)^t f(x, y, a)| \right. \\ + |(\S)^t f(x, y, a) - \pi^{\parvar} f(x, y, a) | \\
\left. + |(\S)^t f(x', y', a') - \pi^{\parvar} f(x', y', a')| \right) = 0, 
\end{multline*}
from which the claim follows. 
\end{proof}

\section{Full algorithm using backward sampling and control variate}
\label{appdx:full:algo}
In this section, we detail Algorithm \ref{alg:full:monty}, which presents one iteration of the online gradient ascent algorithm in the amortized scheme, incorporating both backward sampling and control variates for improved efficiency and variance reduction.

\paragraph{Backward sampling. }
Computing the backward weights of \eqref{eq:backwardweights:online} has the disadvantage of $O(N^2)$ complexity due to the computation of the normalizing constant, which can be prohibitive when $N$ is large (typically for high-dimensional state spaces).
One solution, suggested by \cite{olsson2017efficient} in the context of SMC smoothing, is to use a backward sampling approach. 
More precisely, at time step $t$, given $\sample{t}{i}$, one samples independently $M$ indexes $\{j_k\}_{k = 1}^M$ 
from the categorical distribution over $ \{1,\ldots,N\}$ with weights $\{\nrmbackwdweight{t}{i,j}\}_{j = 1}^N$, and replace \eqref{eq:approx:Hstat:online} by $\sum_{k=1}^M (\approxHstat[\parvar_{t-1}]{t-1}{j_k} + \addfelbo[\parvar_t]{t}(\sample{t-1}{j_k}, \sample{t}{i}))/M$.
\cite{olsson2017efficient} show that even with $M$ much smaller than $N$ (typically, $M=2$), which provides a considerable improvement in complexity, this alternative estimator has only slightly higher variance than the original estimator. 
Here, noting that $\nrmbackwdweight{t}{i,j} \propto_{j} \fwdpot{t}(\sample{t-1}{j}, \sample{t}{i})$, backward sampling can, in the case of bounded potential functions, be performed using an accept-reject procedure without having to calculate the normalizing constant of the weights. 
We refer the reader to  \cite{olsson2017efficient,gloaguen2022pseudo,dau2022complexity} for details and alternative backward sampling approaches. 

\paragraph{Variance reduction of the gradient estimator. }

Proposition \ref{prp:elbo:grad:recursion} involves computing \emph{score-function expectations} in the form $\expect{q^\parvar}{\grad_{\parvar} \log q^\parvar(X) \cdot f(X)}$ for some p.d.f. $q^\parvar$. 
As shown in \cite{mohamed2020}, direct Monte Carlo estimation of the score function leads to high variance and should normally not be used without a suitable variance reduction technique. The most straightforward approach is to design a \textit{control variate}. 
Using the fact that $\mathbb{E}_{q^\parvar}[\grad_{\parvar} \log q^\parvar(X)] = 0$, the target expectation can be rewritten as $\mathbb{E}_{q^\parvar}[\grad_{\parvar} \log q^\parvar(X) \{f(X) - \expect{q^\parvar}{f(X)}\}]$, which can be estimated with lower variance using a Monte Carlo estimate of $\expect{q^\parvar}{f(X)}$. In our case, the latter is formed as a by-product of Algorithm~\ref{alg:onlinenabla}, and  therefore our methodology comes with built-in variance reduction without the need to recompute additional quantities. 
This accelerated version of Algorithm~\ref{alg:onlinenabla} is described in detail in Algorithm~\ref{alg:full:monty} (see Appendix~\ref{appdx:full:algo}), which also includes the backward sampling technique described above.

As an alternative to this variance reduction technique, it is natural to consider the reparametrization trick, as it often leads to Monte Carlo estimators with lower variance compared to those obtained using the score function.
However, the implementation of the reparametrization trick in this context requires that $\grad_{\parvar} \ELBO{t}$ is expressed as an expectation with respect to a random variable $Z_{0:t}$ that does not depend on $\parvar$. Moreover, the recursive expression of this expectation at time $t + 1$ must be derivable from its predecessor, which is non-trivial.
For example, in the classical case where $q_{0:t}$ is the p.d.f. of a multivariate Gaussian random variable with mean $\mu$ and variance $\Sigma$, and the expectation is taken w.r.t. $Z_{0:t} \sim \mathcal{N}(0, I_{\statedim \times (t + 1)})$, such a recursion is not feasible as the ELBO is no longer an additive functional when $X_{0:t}$ is replaced by  $\mu + \Sigma^{\frac{1}{2}}Z_{0:t}$.

\begin{algorithm}[ht] 
\caption{\label{alg:full:monty}One iteration of the online gradient ascent algorithm (for $t\geq 1$) in the amortized scheme}
\begin{algorithmic}
\REQUIRE \textit{}
\begin{itemize}
    \item Previous statistics $\{ (\approxHstat[\parvar_{t-1}]{t-1}{i}, \approxGstat[\parvar_{t-1}]{t-1}{i}, \approxFstat[\thp_{t-1}]{t-1}{i})\}_{i=1}^N$, and previous  samples $\{\sample{t-1}{i}\}_{i=1}^N$;
    \item Intermediate quantity $a_{t-1}$, parameter estimates $(\thp_t, \parvar_{t})$, step sizes $(\gamma^\thp_t,\gamma^\parvar_t)$;
    \item New observation $y_t$.
\end{itemize}
\ENSURE $\{(\approxGstat[\parvar_{t}]{t}{i}, \approxHstat[\parvar_{t}]{t}{i})\}_{i=1}^N$, $\parvar_{t + 1}$, $a_t$.
\STATE
Set $a_t \gets \amap{\parvar_t}(a_{t-1}, y_t)$\ and $\eta^{\parvar_t}_t \gets \fat[t](a_{t})$, the parameters of $q_t^{\parvar_{t}}$ 
\STATE Sample $\{\sample{t}{i}\}_{i=1}^N$ independently from $\vd[\parvar_{t}]{t}$\;
\FOR{$i \gets 1$ to $N$}
\STATE Set $\tilde{\eta}_{t}^{\parvar_{t},i} \gets \fxt[t](\sample{t}{i})$\;
\FOR{$j \gets 1$ to $M$}
\STATE  \textit{// Backward sampling step, $M$ is the number of backward samples} 
\STATE Sample $(j_k)_{k = 1}^M\overset{\text{i.i.d.}}{\sim} \mathsf{Cat}(\{\nrmbackwdweight[\parvar_t]{t}{i,j}\}_{1\leq j \leq N})$ with the weights  \eqref{eq:backwardweights:online}. 
\ENDFOR
\STATE Set  \textit{// Recall that each term $\addfelbo[\parvar_t]{t}$ depends on $y_t$.}
\begin{align*}
\approxHstat[\parvar_{t}]{t}{i} &\gets 
\frac{1}{M}\sum_{k=1}^M
\left(
\approxHstat[\parvar_{t - 1}]{t - 1}{j_k} + \addfelbo[\parvar_{t}]{t}(\sample{t-1}{j_k},\sample{t}{i})\right) \eqsp;\\
\approxGstat[\parvar_{t}]{t}{i} &\gets 
\frac{1}{M}\sum_{k=1}^M
\left\{
\approxGstat[\parvar_{t - 1}]{t-1}{j_k} + \grad_\parvar \log \vd[\parvar_{t}]{t-1|t} (\sample{t-1}{j_k},\sample{t}{i})\left( \approxHstat[\parvar_{t-1}]{t - 1}{j_k} + \addfelbo[\parvar_{t}]{t}(\sample{t-1}{j_k},\sample{t}{i}) - \approxHstat[\parvar_{t}]{t}{i} \right)\right\}\eqsp;\\
\approxFstat[\thp_t]{t}{i} &\gets \sum_{i=1}^N \nrmbackwdweight[\parvar_t]{t}{i,j}\left\{\approxFstat[\thp_{t-1}]{t - 1}{j} + \grad_\thp \addfelbo[\parvar_t]{t}(\xi_{t-1}^j, \xi_t^i) \right\}\eqsp;
\end{align*}
\STATE  \textit{// Note the difference with \eqref{eq:approx:Gstat:online} and the inclusion of control variate $\approxHstat{t}{i}$ for the computation of $\approxGstat{t}{i}$} 
\STATE \textit{// }$\grad\vd[\parvar_{t}]{t-1|t} (\sample{t-1}{j_k},\sample{t}{i})$ \textit{is typically computed with automatic differentiation}
\ENDFOR
\STATE Set 
\begin{align*}
\hGELBOt{\parvar}{t}{t} &\gets \frac{1}{N}\sum_{i=1}^N
\left\{\approxGstat[\parvar_{t}]{t}{i} + \grad_\parvar \log \vd[\parvar_{t-1}]{t}(\sample{t}{i})\left(\approxHstat[\parvar_{t}]{t}{i} - \frac{1}{N}\sum_{k=1}^N \approxHstat[\parvar_{t}]{t}{k}\right)\right\} \ ;\\
\parvar_{t + 1} &\gets \parvar_{t} + \gamma^{\parvar}_t \left( \hGELBOt{\parvar}{t}{t} - \hGELBOt{\parvar}{t-1}{t-1}\right) \; \\
\hGELBOt{\thp}{t}{t} &\gets \frac{1}{N}\sum_{i = 1}^{N}\approxFstat[\thp_t]{t}{i}\,;\\
\thp_{t+1} &\gets \thp_t + \gamma^{\thp}_{t+1}\left(\hGELBOt{\thp}{t}{t} - \hGELBOt{\thp}{t-1}{t-1}\right)\eqsp.
\end{align*}
\STATE \textit{// Note the difference with Algorithm \ref{alg:onlinenabla} and the inclusion of the control variate $\frac{1}{N}\sum_{i=1}^N \approxHstat{t}{i}$ in the calculation of $\widehat{\grad} \ELBO{t}$}
\STATE \textit{// }$\grad \log \vd[\parvar_{t-1}]{t}(\sample{t}{i})$ \textit{is typically computed using automatic differentiation}
\end{algorithmic}
\end{algorithm}

\section{Supplementary details for the numerical experiments in \Cref{sec:exp}}
\label{appdx:experiments}

\subsection{Appendix for section \ref{sec:LGHMM} the linear Gaussian SSM.} \label{sec:details:lin:Gauss}

\paragraph{Variational Family.}
In the linear Gaussian setting, the variational marginals $\vd{t}$ are parameterized as Gaussian distributions defined by their natural parameters $\eta_t^\parvar$. To ensure the backward kernel $\vd{t-1|t}$ remains in the same Gaussian family, the potential is explicitly defined as 
\begin{equation}
    \fwdpot{t}(x_{t-1},x_t) = \exp{(\langle \tilde{\eta}_{t}^{\parvar}(x_t), T(x_{t-1}) \rangle)} \eqsp.
\end{equation}
This formulation allows the backward density to be derived analytically by simply summing natural parameters, avoiding the need for normalizing constants. The parameters for the backward kernel are updated according to:
\begin{equation}
    \eta_{t-1|t}^\parvar = \eta_{t-1}^\parvar + \tilde{\eta}_t^\parvar \eqsp.
\end{equation}
Unlike the general case requiring neural networks, the recursions for these parameters in the linear Gaussian case are analytical, effectively mirroring the smoothing distribution updates of a standard linear Gaussian SSM. 

\paragraph{Parameters for the linear Gaussian SSM.}
For the streaming experiment presented in the main text ($T=50,000$), we learn both model and variational parameters from random initialization. We utilize learning rates of $10^{-3}$ for the variational parameters and $10^{-4}$ for the model parameters. 
To ensure a rigorous comparison, we replicate the generative settings of \cite{campbell2021online}, using diagonal noise covariance matrices with fixed variances of $0.1$ for the transition and $0.25$ for the emission. 

\paragraph{Oracle ELBO} As an oracle baseline, we can compute the closed-form ELBO and its associated gradient via the reparameterization trick. Figure \ref{fig:training_curve_lgm} displays the evolution of the ELBO in the case of the linear SSM and the offline setting, \emph{i.e.}, when observations are processed through multiple epochs. In this specific experiment, we do not learn the model parameter. 
For our recursive method, we choose $\Delta=2$ to truncate the backpropagation, as we observe that $\Delta < 2$ prevents our method from converging altogether, while $\Delta > 2$ only improves convergence speed by a small margin. 
The experiment is run using $10$ different parameters for the generative model, $\statedim=\obsdim=10$, $T=500$, and $N=2$ for the two methods involving Monte Carlo sampling.
It shows the convergence of our score-based solution to the correct optimum given by the analytical computations. 
This is particularly appealing and notably demonstrates that our online gradient-estimation method may perform well using few samples. 
In practice, we observe that the variance reduction introduced in Section \ref{appdx:full:algo} is crucial in reaching such performance.

\begin{figure}[t]
     \centering
     \includegraphics[width = 0.5\textwidth]{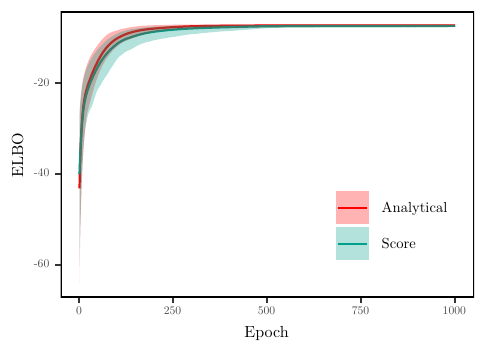}
     \caption{Evolution of $\ELBO{T}/T$ computed with three different methods and with three different types of gradients estimates. Full lines: means of the 10 replicates. Shaded lines: standard deviations of the 10 replicates.}
     \label{fig:training_curve_lgm}
 \end{figure}

\begin{table}[t]
    \small
    \centering
    \begin{tabular}{||c||c|c||} 
        \hline
        Gradients & $\Delta_{T,p}^\parvar$ ($\times 10^{-2}$) & Avg. time \\ [0.5ex] 
        \hline\hline
        Score-based & 13.5 $\pm$ 0.7 ({\bf 12.2}) & 173 ms \\ 
        \hline
        Backward sampling & 11.9 $\pm$ 0.4 ({\bf 11.4}) & 17 ms \\
        \hline
    \end{tabular}
    \caption{RMSE of the predicted marginal means $\expect{\vd{0:T}}{X_t}$ w.r.t. the true states $x_t^*$ and average time per gradient step.}
\label{table:smoothing_rmse_chaotic_dim_5}
\end{table}

\subsection{Appendix for section \ref{sec:xp:chaotic:rnn} the chaotic RNN}
\label{appdx:chaotic:rnn}
The 1-step smoothing and filtering errors in table \ref{table:1_step_smoothing_chaotic_rnn} are given as 
\begin{align}\label{eq:1-step_error}
\kappa^{(1)}_{T} &= \frac{1}{T-1}\sum_{t=1}^{T-1} \Big(\tfrac{1}{\statedim}\sum_{k=1}^{\statedim} \left(\approxexpect{\vd{t-1:t}}{X_{t-1}^{(k)}} - x_{t-1}^{*(k)}\right)^2\Big)^{\!1/2}\\
\kappa^{(2)}_{T}&=\frac{1}{T}\sum_{t=1}^{T} \Big(\tfrac{1}{\statedim}\sum_{k=1}^{\statedim} \left(\approxexpect{\vd{t}}{X_{t}^{(k)}} - x_{t}^{*(k)}\right)^2\Big)^{\!1/2}\eqsp\,,
\end{align}
where $\widehat{\mathbb{E}}_q$ denotes the standard Monte Carlo estimate of an expectation w.r.t. $q$.
Table~\ref{table:chaotic_streaming_data1} reports smoothing and filtering RMSE with respect to the true states at the end of training.

\begin{table}[t]
    \centering
    \small
    \begin{tabular}{||c||c|c||} 
        \hline
        Sequence & Smoothing RMSE & Filtering RMSE  \\ [0.5ex] 
        \hline
        Training & 0.281 & 0.311 \\ 
        \hline
        Eval & 0.278 ($\pm$ 0.01) & 0.305 ($\pm$ 0.014) \\
        \hline
    \end{tabular}
    \caption{Smoothing and filtering RMSE when $\parvar$ is learned online together with $(\rho,\gamma)$ in the chaotic RNN. Results are shown for the training stream and for independent sequences from the same generative model.}
    \label{table:chaotic_streaming_data1}
\end{table}

\paragraph{Parameters for the chaotic RNN}
We choose the same generative hyperparameters as \cite{campbell2021online} with $\Delta=0.001$, $\tau=0.025$, $\gamma=2.5$, $2$ degrees of freedom and a scale of $0.1$ for the Student-$t$ distribution, and define $Q = \text{diag}(0.01)$. 
For the joint learning setting, training is performed using the Adam \cite{kingma2015adam} optimizer with learning rates of $10^{-3}$ for the variational parameters and $10^{-4}$ for the model parameters.
The variational backward kernels follow the parameterization defined in the main text, where the potential functions are parameterized by MLPs with $100$ hidden units and $\tanh$ activations.

\paragraph{Implementation settings for the comparison with \cite{campbell2021online}}

In the original paper, the variational distribution is designed in a non-amortized scheme, meaning that the variational parameters are not shared through time.
Specifically, we have that $\parvar = \lbrace \parvar^0, \dots, \parvar^t,\dots, \rbrace$, and each $\parvar^t$ contains the parameter $\eta_t = (\mu_t, \Sigma_t)$ of the distribution $\vd[\parvar]{t} \sim \mathcal{N}(\mu_t, \Sigma_t)$ and the parameter $\tilde{\eta}_t$ of the function $\fwdpot[\parvar_t]{t}$. 
For this latter function, we match the number of parameters of \cite{campbell2021online} by defining $\fwdpot[\parvar]{t}(x,y) = \exp{(\tilde{\eta}_t(y) \cdot T(x))}$ with $\tilde{\eta}_t(y) = (\tilde{\eta}_{t,1}(y), \tilde{\eta}_{t,2})$ where $y \mapsto \tilde{\eta}_{t,1}(y)$ is a multi-layer perceptron with 100 neurons from $\statesp$ to $\statesp$, and $\tilde{\eta}_{t,2}$ is a negative definite matrix. We follow the optimization schedules of \cite{campbell2021online} with $K=500$ gradient steps at each time-step. 

\paragraph{Modifications induced by the non-amortized scheme}
\label{appdx:details:non:amortized}
In the non-amortized scheme, $\parvar = \lbrace \parvar^0, \dots, \parvar^t\rbrace$ is a set of distinct parameters, each parameter corresponding to a specific time step.
In the notations of the article, the estimate $\parvar_{t-1}$ of $\parvar$ after having processed observations the $y_{0:{t-1}}$ depends on $\lbrace\parvar^0, \dots, \parvar^{t-1}\rbrace$. 
Therefore, the gradient of the ELBO  will be w.r.t. $\parvar^t$ only. 
This affects the expression of the statistic $\Gstat{t}$, and one can see in the expansion of Eqn.~ \eqref{eq:appdx:beg:G} that the term \eqref{eq:appdx:Gtm1} will now, when the gradient is taken w.r.t. $\parvar^t$, be zero.  
This means that this term no longer has to be propagated.
Indeed, as we set $\vd[\parvar_t]{t-1|t}(x_t, x_{t-1}) \propto \vd[\parvar_{t-1}]{t - 1}(x_{t-1})\fwdpot[\parvar_t]{t}(x_{t-1}, x_t)$, the gradient of the ELBO w.r.t. $\parvar^t$ will be
\begin{align*}
\gradalt{\parvar^t} \ELBO[\parvar_t]{t} =& \mathbb{E}_{\vd[\parvar_t]{t}}
\left[ 
\mathbb{E}_{\vd[\parvar_t]{t - 1\vert t}}
\left[
\grad \log \vd{t}(X_t) \times \addfelbo[\parvar_t]{t}(X_{t-1}, X_t) \right.
\right.
\\
&+ 
\left.\left.\gradalt{\parvar^t} \log \vd[\parvar_t]{t - 1\vert t}(X_{t-1}, x_{t})\times \left(\Hstat{t-1}(X_{t-1}) + \addfelbo[\parvar_t]{t}(X_{t-1}, X_t) \right)\right]\right]\eqsp.
\end{align*}

This gradient will be estimated using Monte Carlo in the same way as in the algorithm. In \cite{campbell2021online}, the inner conditional expectation is estimated with the regression approach (that we briefly recall below) instead of our importance sampling approach. 

\paragraph{Functional regression approach of \cite{campbell2021online}} 
Here, we recall the alternate option used in \cite{campbell2021online} to propagate approximations of the backward expectations.
Denoting $\mathcal{F} = \left\lbrace g: \rset^p \to \rset^{\statedim}, \mathbb{E}_{\vd{t}}[\|g(X_t)\|_2] < \infty\right\rbrace$,  $\Hstat{t}(x)$ satisfies, by definition of conditional expectation, 
$$
\Hstat{t} = \operatorname*{argmin}_{g \in \mathcal{F}} \mathbb{E}_{\vd{t-1:t}(x_{t-1}, x_t)} \|g(X_t) -  [\Hstat{t-1}(X_{t-1})  + \addfelbo{t}(X_{t-1}, X_t)]\|_2,
$$
which provides a regressive objective for learning an approximation of $\Hstat{t}$.
In practice, the authors restrict the minimization problem to a subset of $\mathcal{F}$, a parametric family of functions (typically, a neural network) parameterized by $\gamma$ belonging to $\Gamma \subset \rset^{d_\gamma}$, and learn this by approximating the expectation using Monte Carlo sampling. 
More precisely, the authors propose to estimate $\Hstat{t}$ by $\paramapproxvarcondE{t}$, where 
\begin{equation}
\label{eq:campbell:minimization}
\hat{\gamma}_t = \operatorname*{argmin}_{\gamma \in \Gamma}
\frac{1}{N}\sum_{k = 1}^N \|
\paramvarcondE(\xi^k_t) - 
[\paramapproxvarcondE{t-1}(\xi^k_{t-1})  + \tilde{h}_t(\xi^k_{t-1}, \xi^k_t)]\|_2\eqsp,
\end{equation}
where $\left\lbrace (\xi_{t-1}^i, \xi_{t}^i) \right\rbrace_{i=1}^N$ is  an i.i.d. sample from the variational joint distribution of $(X_{t-1}, X_t)$, which has density $\vd{t-1:t} = \vd{t}\vd{t-1|t}$.
Upon convergence, $\paramapproxvarcondE{t}$ is then used in the successive recursions (similar to \eqref{eq:approx:Hstat:online}). 

\subsection{Appendix for section \ref{sec:exp_air_quality}}
\label{appdx:air_quality}
\paragraph{Air-quality data}
The UCI Air Quality dataset contains hourly readings from 5 metal oxide chemical sensors located in a significantly polluted area. 
Our observation vector $y_t$ comprises Carbon Monoxide (CO), Non-Methane Hydrocarbons (NMHC), Nitrogen Oxides ($\text{NO}_x$, $\text{NO}_2$), Benzene ($\text{C}_6\text{H}_6$), Temperature (T), Relative Humidity (RH), and Absolute Humidity (AH).
A key challenge of this dataset is the non-uniform data censorship caused by sensor failures. 
As shown in Figure~\ref{fig:air_quality_panels}, these dropouts occur at different intervals for different sensors. 
In the raw data, these are marked by a sentinel value ($-200$), which we treat as missing values (NaNs) during preprocessing.

\begin{figure}[h]
    \centering
    \includegraphics[width=0.95\textwidth]{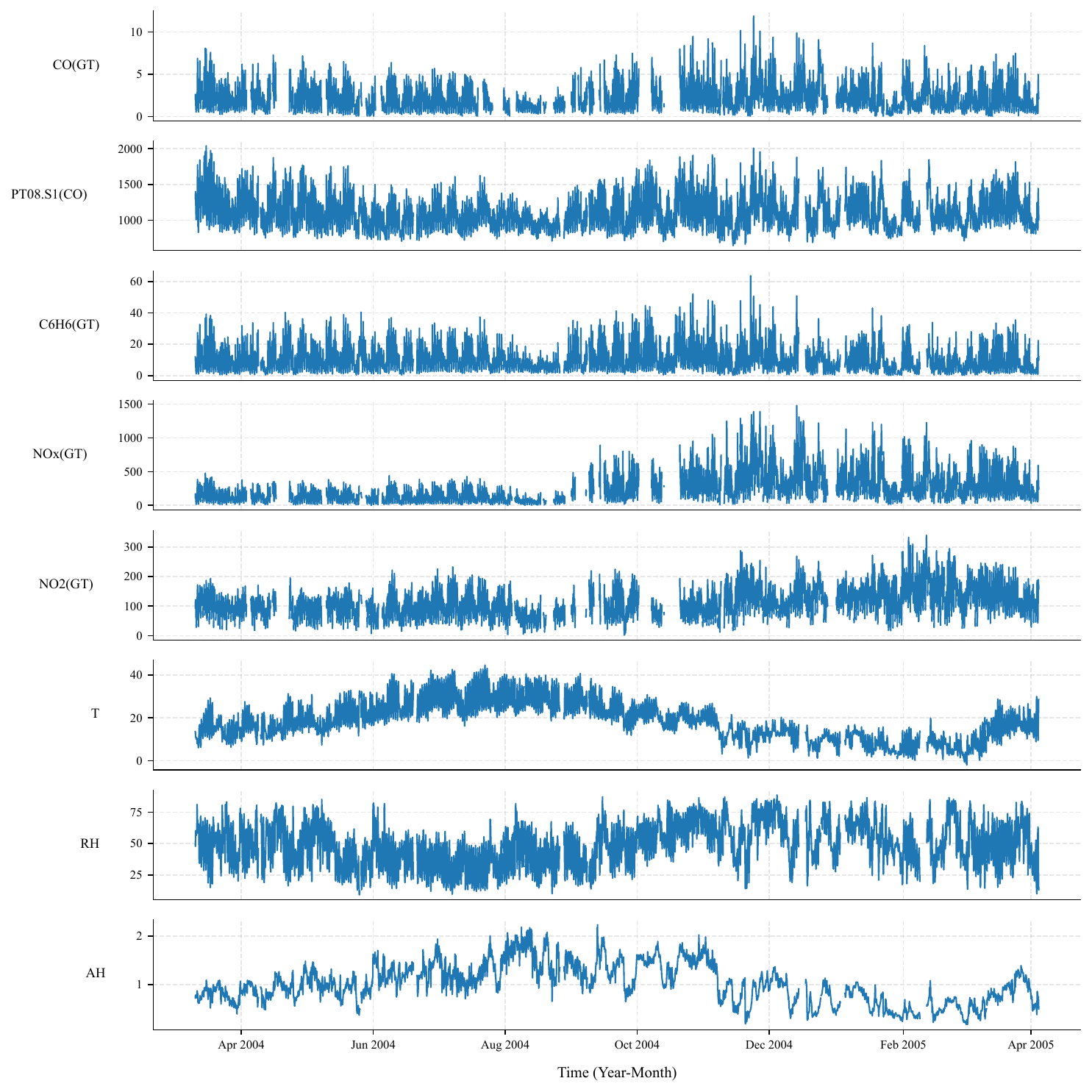}
    \caption{Time series of the 8 observed dimensions in the UCI Air Quality dataset. The gaps in the signals correspond to periods of sensor censorship (missing data). The framework must infer the latent state during these blackout periods without direct observation.}
    \label{fig:air_quality_panels}
\end{figure}

\paragraph{Parameters for the Air Quality Data}
The transition function $f_\theta$ and emission function $g_\theta$ are both parameterized by MLPs with 2 hidden layers of 32 units each and $\tanh$ activation functions. The noise covariances $Q_\theta$ and $R_\theta$ are learned diagonal matrices.
Regarding the variational approximation, the backward potentials are parameterized by MLPs with 32 hidden units and $\tanh$ activations.
We use $N=20$ importance samples for the Monte Carlo estimates.
Training is performed using the Adam optimizer with learning rates of $10^{-3}$ for the variational parameters and $10^{-4}$ for the model parameters.

\end{document}